\documentclass[aps,pra,reprint,showkeys]{revtex4-1}
\usepackage{amsmath,paralist,amsthm,comment,amssymb}
\usepackage{times,color}
\usepackage{epsfig}
\usepackage{amscd}
\usepackage{tikz-cd}
\usepackage{graphicx,float}
\graphicspath{ {./Figures/} }
\usepackage{caption}
\usepackage{subcaption}
\usepackage{empheq}
\usepackage{appendix}
\usepackage{tikz}
\usetikzlibrary{matrix,arrows}
\usetikzlibrary{shapes.geometric}
\usetikzlibrary{patterns} 

\usepackage{multirow}

\usepackage{pgfplots}

\usepackage{comment}

\usepackage[breaklinks=true,colorlinks=true,linkcolor=blue,urlcolor=blue,citecolor=red]{hyperref}

\usepackage{float}
\makeatletter
\let\newfloat\newfloat@ltx
\makeatother

\usepackage{algorithm}
\usepackage{algpseudocode}
\algrenewcommand\algorithmicrequire{\textbf{Input:}}
\algrenewcommand\algorithmicensure{\textbf{Output:}}

\usepackage{datetime}

\long\def\/*#1*/{}

\newcommand{\nix}[1]{}

\newcommand{\ket}[1]{|#1\rangle}

\newtheorem{theorem}{Theorem}

\newtheorem{lemma}[theorem]{Lemma}
\newtheorem{proposition}[theorem]{Proposition}

\newtheorem{remark}[theorem]{Remark}

\newtheorem{definition}[theorem]{Definition}

\begin{document}
\title{Color codes with twists: construction and universal gate set implementation}
\author{Manoj G. Gowda}
\author{Pradeep Kiran Sarvepalli}
\affiliation{
Department of Electrical Engineering, Indian Institute of Technology Madras, Chennai 600 036, India
}

\begin{abstract}
Twists are defects in the lattice that can be used to perform encoded computations.
Three basic types of twists can be introduced in color codes, namely, twists that permute color, charge of anyons and domino twists that permute the charge label of an anyon with a color label.
In this paper, we study a subset these twists from coding theoretic viewpoint.
Specifically, we discuss systematic construction of charge permuting and color permuting twists in color codes.
We show that by braiding alone, Clifford gates can be realized in color codes with charge permuting twists.
We also discuss implementing single qubit Clifford gates by Pauli frame update and CNOT gate by braiding holes around twists in color codes with color permuting twists.
Finally, we also discuss implementing a non-Clifford gate by state injection, thus completing the realization of a universal gate set.
\end{abstract}
\maketitle

\section{Introduction} 
\subsection{Motivation}
Color codes are interesting from the point of view of fault tolerant quantum computation.
In two dimensions they allow transversal implementation of Clifford gates \cite{Bombin2006}.
There are various ways of doing encoded computation using 2D color codes.
One of the early approaches was due to Fowler \cite{Fowler2011}.
In this information was encoded using holes and gates realized using code deformation. 
Independently, Landahl et al.~\cite{Landahl2011} also proposed alternate protocols for fault tolerant quantum computation with color codes. 
They also used  holes for encoding information, 
and code deformation for implementing encoded gates.

Lattice surgery~\cite{Horsman2012,Landahl2014, Litinski2019} is an alternative to quantum computation with holes. 
Lattice surgery was used to show universal computation in color codes and the analysis of the resources required for the same was studied by Landahl and Ryan-Anderson~\cite{Landahl2014}.

Twists are yet another method to encode information into a 2D lattice. 
Furthermore, they can also be used to perform encoded gates. 
Twists are defects in the lattice that permute the label of anyons.
Loosely speaking, an anyon is the syndrome resulting from a violated stabilizer.
Anyons in color codes have far richer structure compared to surface codes thus allowing for more permutations.
Anyons in color codes are characterized by two labels, namely, charge and color.
Twists in color code can permute charge, color or both.
Kesselring et al. \cite {Kesselring2018} classified  the various types of twists possible in color codes. 
However, certain aspects of color codes with twists were left unexplored. 
This motivates us to undertake a study of twists in color codes and their application to fault-tolerant quantum computation.
Another reason for our study comes from the fact that twists can potentially lead to lower complexity for quantum computation. 
For instance,  twists in surface codes are shown to provide gains in space time complexity of computation~\cite{Hastings2015}.

Braiding in surface codes needs both lattice modification as well as stabilizer modification.
A particular type of twist in color codes, viz. charge permuting twist, requires no lattice modification for movement.
Only the stabilizers of the twist face and the faces in the path of the twist movement need to be changed. 
This can be accomplished without lattice modification for charge permuting twists. 
Therefore, braiding procedure is simpler in comparison to surface codes with twists.

\subsection{Previous work and contributions}
The use of dislocations for encoding was suggested by Kitaev~\cite{Kitaev2003}.
The ends of a dislocation are called twists.
Twists in toric code model were first studied by Bombin~\cite{Bombin2010} where twists were shown to behave like non-Abelian anyons.
Twists were studied from the perspective of unitary tensor categories by Kitaev and Kong~\cite{Kitaev2012}.
Yu and Wen~\cite{Yu-Wen2012} studied twists in qudit systems  where the authors showed that twists exhibit projective non-Abelian statistics.
Hastings and Geller~\cite{Hastings2015}  showed that by using twists in surface code along with arbitrary state injection leads to reduction in amortized time overhead.
Twists in surface codes were identified with corners of lattice by Brown et al.~\cite{Brown2017} which helps in performing single qubit Clifford gates.
Further, the authors also proposed a hybrid encoding scheme with holes and twists in surface codes that was made use of to implement CNOT gate.
Surface codes with twists in odd prime dimension and braiding protocols to implement generalized Clifford gates were studied in Ref.~\cite{GowdaSarvepalli2020}.

Twists have also been studied in other class of topological codes.
Bombin~\cite{Bombin2011} studied twists in topological subsystem color codes and showed that Clifford gates can be implemented by braiding twists. 
Litinski and von Oppen~\cite{Litinski2018}  studied twists in Majorana surface codes and have shown that all logical Clifford gates can be done with zero time overhead.
Kesselring et al.~\cite{Kesselring2018} have cataloged the twists in color codes.
They have presented the basic twists with which all twists in color codes can be realized.
We build upon their work, and explore in detail two types of twists in color codes namely, the charge permuting twists and color permuting twists. 
The work of Kesselring et al.~\cite{Kesselring2018} approached color codes with twists from
the abstract theory of boundaries and domain walls. 
In contrast, in this paper we 
take a more concrete approach giving explicit constructions and protocols for quantum computation. 
Our contributions are listed below.
\begin{compactenum}[(i)]
\item We propose systematic constructions of color codes with charge permuting twists and color permuting twists from a $2$-colex.
We focus on $X$ type charge permuting twists. 
while Ref.~\cite{Kesselring2018} focussed on $Y$ type charge permuting twists.
The color permuting twists we study have  different lattice representation than the ones studied in 
Ref.~\cite{Kesselring2018}. This also results in a different string algebra for the  Pauli operators. 
\item We propose an implementation of Clifford gates by braiding charge permuting twists. 
This result is stated in Theorem~\ref{thm:single-qb-Clifford} for single qubit Clifford gates and in Theorem~\ref{thm:multi-qb-entangling} for controlled-$Z$ gate up to a phase gate on control and target qubits.
Ref.~\cite{Kesselring2018} did not provide explicit details of these protocols while they discussed code deformation using twists for implementing quantum gates. 

\item We propose an implementation of single qubit Clifford gates by Pauli frame update and CNOT gate by braiding holes around color permuting twists.
The latter protocol is inspired by the work of Brown et al.~\cite{Brown2017} who proposed using a combination of holes and twists for implementing CNOT gate in the context of surface codes. 
\end{compactenum}
\smallskip

Finally, we discuss the implementation of a non Clifford gate for both types of twists.
The protocol is a minor variation of the one presented in Ref.~\cite{Bravyi2006} and could be of independent interest.
A list of our main results is given in Table~\ref{tab:imp_lm-thm}.

\begin{table}[htb]
    \centering
    \begin{tabular}{p{2.5cm}|p{2.5cm}|p{2.5cm}}
    \hline
    \hline
     Contribution / Twist type & Charge permuting twists & Color permuting twists\\
     \hline
     Code construction & Lemma~\ref{lm:encoded-qubits-charge} & Lemma~\ref{lm:encoded-qubits-cp}\\
     \hline
     Gates & Theorem~\ref{thm:single-qb-Clifford} and Theorem~\ref{thm:multi-qb-entangling} & Theorem~\ref{thm:color-perm-gates} \\
     \hline
    \end{tabular}
    \caption{A list of main results of the paper.}
    \label{tab:imp_lm-thm}
\end{table}

Our treatment of the color codes with twists makes extensive use of mappings between  Pauli operators and strings on the code lattices. 
This approach brings out the geometry of the stabilizer generators and the logical operators in a very transparent fashion. 
Furthermore, it allows for a consistent and simplified analysis of the encoded gates using these codes.

\subsection{Overview}
In Section~\ref{sec:background}, we briefly discuss color codes and review the work presented in Ref.~\cite{Kesselring2018}.
In Section~\ref{sec:charge}, we study charge permuting twists.
while  Section~\ref{sec:color} is devoted to color permuting twists.
We discuss code parameters, logical operators and string representation for Pauli operators in the presence of color permuting twists.
In Section~\ref{sec:gates-charge-permuting}, we show that encoded Clifford gates can be realized by braiding charge permuting twists.
In Section~\ref{sec:gates-color-permuting}, we show that Clifford gates can be realized using charge permuting twists by Pauli frame update and joint parity measurements.
\newline

\noindent \emph{Notation.}
Some frequently used notation in the paper is given in Table~\ref{tab:notation}.

\begin{table}[H]
    \centering
    \begin{tabular}{c|c}
    \hline
    \hline
        Notation & Meaning \\
        \hline
        $\mathsf{F}_c$, $c \in \{r,g,b \}$ & Faces of color $c$ \\
        \hline
        $\mathsf{F}_{c c^\prime}$ & Faces with colors $c$ or $c^\prime$ i.e., $\mathsf{F}_c\cup \mathsf{F}_{c'}$ \\
        \hline
        $V(f)$ & Vertices of a face $f$\\
         \hline
         $A_2(f)$ & Faces that share an edge with the face $f$ \\
         \hline
         $\mathcal{W}_{i,j}^{c}$ & String of color $c$ encircling twists $t_i$ and $t_j$ \\
         \hline
    \end{tabular}
    \caption{Notations used in the paper and their meaning.}
    \label{tab:notation}
\end{table}

\section{Background}
\label{sec:background}
We briefly review the necessary background on 2D color codes, see \cite{Bombin2006} for more details.
Color codes are defined on trivalent and three-face-colorable lattices embedded on a two dimensional surface. Such lattices are called $2$-colexes. 
A well known example of trivalent and three colorable lattice is the honeycomb lattice. 
Two stabilizers generators are defined on every face of the lattice:
\begin{equation}
B_f^{X} = \prod_{v \in V(f)}X_v, \text{ and } B_f^{Z} = \prod_{v \in V(f)} Z_v,
\label{eqn:cc-stabilizers}
\end{equation}
where $V(f)$ denotes the vertices that belong to the face $f$ and $X$ and $Z$ are the Pauli operators.
The stabilizers defined in Equation~\eqref{eqn:cc-stabilizers} apply to faces of all colors. 
However, not all stabilizers are independent. 
They satisfy the following relations~\cite{Bombin2006},
\begin{subequations}
\begin{eqnarray}
\prod_{f \in \mathsf{F}_r} B_f^X &=& \prod_{f \in \mathsf{F}_g} B_f^X = \prod_{f \in \mathsf{F}_b} B_f^X, \\
\prod_{f \in \mathsf{F}_r} B_f^Z &=& \prod_{f \in \mathsf{F}_g} B_f^Z = \prod_{f \in \mathsf{F}_b} B_f^Z.
\end{eqnarray}
\label{eqn:cc-stab-constraint}
\end{subequations}
Here  $\mathsf{F}_c$ denotes the set of faces of color $c \in \{ r, g, b\}$.
Equations~\eqref{eqn:cc-stab-constraint} indicate that there are four dependent stabilizers.
The number of logical qubits depends on the topology of surface on which graph is embedded.
If the surface has genus $g$, then the number of encoded qubits is $4g$.
The genus of  a surface is half the number of independent non-contractible loops that can be drawn on the surface.
For instance, any loop drawn on a sphere can be contracted to a point and hence its genus is zero.
On a torus, one can draw two independent non-contractible loops and hence its genus is $1$.

Of particular interest to us in this paper is the case $g = 0$.
Examples of such lattices are shown in Fig.~\ref{fig:lattice-boundary}.
Two stabilizers, $X$ type and $Z$ type, are defined for each  face of the lattice.
All vertices in the lattice are trivalent and hence we have $3v = 2e$, where $e$ and $v$ are the number of edges and vertices of the lattice respectively.
Using this in the Euler characteristic equation $v + f - e = 2$, where $f$ is the number of faces, we get $2f = v + 4$.
It is easy to verify that this lattice also obeys the constraints given in Equations~\eqref{eqn:cc-stab-constraint}.
Therefore, the number of independent stabilizers is $s = 2f-4 = v$, thus giving zero encoded qubits.

\begin{figure}[htb]
    \centering
    \begin{subfigure}{.225\textwidth}
    \centering
    \includegraphics[scale = 1.25]{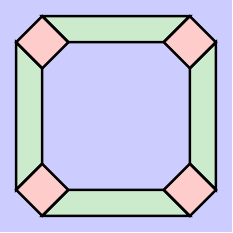}
    \subcaption{}
    \label{fig:sq-octagon-lattice}
    \end{subfigure}
~
\begin{subfigure}{.225\textwidth}
    \centering
    \includegraphics[scale = .5]{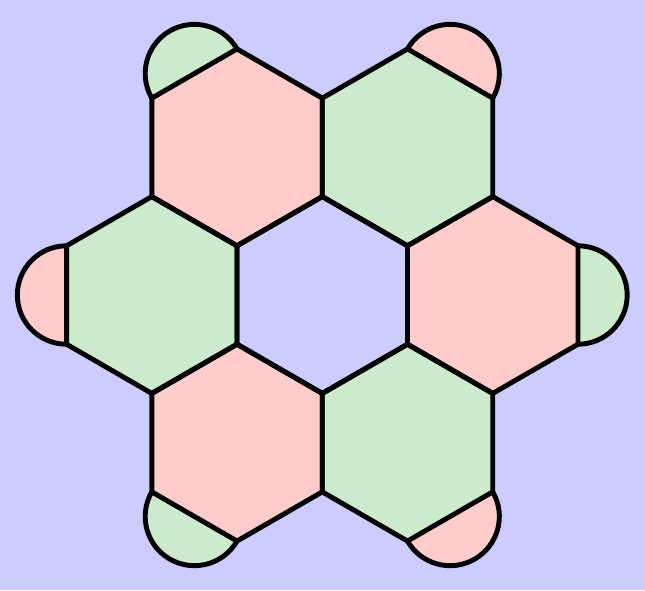}
    \subcaption{}
    \label{fig:hexagon-lattice-boundary}
    \end{subfigure}
    \caption{Trivalent and three colorable lattices with boundary. (a) Square octagon lattice that does not encode any logical qubits. The outer unbounded face has blue color. (b) Hexagon lattice that does not encode any logical qubits. The outer unbounded face has blue color. Stabilizers along the boundary are of weight two.}
    \label{fig:lattice-boundary}
\end{figure}

\subsection{Previous work on twists in color codes}
We now review some relevant material from Ref.~\cite{Kesselring2018}.
The syndromes in color codes can be written as $cp$ where $c \in \{r,g,b\}$ denotes the color of face on which the syndrome is present and $p \in \{x,y,z\}$ denotes the type of stabilizer violated.
If a $Y$ error occurs, then both $X$ and $Z$ stabilizers are violated which is indicated as $cy$.
The syndromes present in the color code are given in Table~\ref{tab:bosons}~\cite{Kesselring2018}.

\begin{table}[htb]
\centering
\begin{tabular}{c|c|c}
\hline
$rx$  & $gx$ & $bx$  \\ \hline
$ry$ & $gy$ & $by$ \\ \hline
$rz$ & $gz$ & $bz$\\\hline
\end{tabular}
\caption{Types of syndromes present in  color codes. The color label of a syndrome viz., $r$, $g$, $b$ denotes the color of the face on which the syndrome is present and the charge label $x$, $y$, $z$ indicate the type of stabilizer violated.}
\label{tab:bosons}
\end{table}

\textit{Twists} are defects in the lattice that permute a label of an anyon to another i.e. they permute anyons when anyons go around twist (more precisely, when they cross the domain wall).
\textit{Domain wall} is a virtual path in the lattice that marks the point where anyons are permuted.
In color codes with twists, there is more than one way to permute anyons and the corresponding twists are given below.

\begin{compactenum}[a)]
\item \textit{Charge permuting twists}. 
These twists exchange the Pauli label of a pair of anyons in a column of Table~\ref{tab:bosons} while leaving the third unchanged.
For example, $x$ and $z$ labels are exchanged and $y$ label is left unchanged.
These twists do not permute the color.
For example, $rx$ is permuted to $rz$ and vice versa but $ry$ is unaltered.

\item \textit{Color permuting twists}. 
These twists correspond to exchange of the color label of a pair of anyons in a row of 
Table~\ref{tab:bosons}.
They leave the color of the third entry in that row unchanged while exchanging the colors of the other two entries. 
Note that these twists permute only color but not charge.
For example a twist that leaves red anyons unchanged but permute $b\alpha$ to $g \alpha$ and vice versa.
This corresponds to exchanging the corresponding entries in two rows in Table~\ref{tab:bosons}.

\item \textit{Domino twists}. Apart from exchanging anyons along rows and columns, one can also exchange them across the diagonal in Table~\ref{tab:bosons}.
In this case, the twist acts by transposition i.e. the diagonal entries $rx$, $gy$ and $bz$ are unaltered whereas the following permutations take place: $ry \leftrightarrow gx$, $rz \leftrightarrow bx$ and $gz \leftrightarrow by$.
This transformation can also be seen as permutation of color and charge labels: $r \leftrightarrow x$, $g \leftrightarrow y$ and $b \leftrightarrow z$.
The anyons $rx$, $gy$ and $bz$ contain the charges that are mutually permuted and hence are left invariant whereas the other bosons are permuted. 
For example, under this permutation $ry$ will be permuted to $bx$ which is the diagonally opposite entry in Table~\ref{tab:bosons}.
\end{compactenum}

\begin{remark}
\label{rmk:other-twists}
It is to be noted that these are fundamental twist types.
One can combine two or more of these twist types to obtain other twists.
\end{remark}

\begin{table}[htb]
    \centering
    \begin{tabular}{p{2cm}|p{2cm}|p{2cm}|p{2cm}}
    \hline
    \hline
       Features / Twist type & Charge permuting  & Color permuting  &  Domino\\
         \hline
        Lattice modification  & Not needed & Needed & Needed \\
         \hline
         Physical qubits added or removed from lattice  & No & Removed &  Added\\
         \hline
         Type of code & non-CSS & non-CSS${}^\dagger$ & non-CSS \\
         \hline
         Geometry of logical operators & Closed strings encircling a pair of twists & Closed strings encircling a pair of twists & Closed strings encircling a pair of twists\\
         \hline
         $T$-line & Absent & Present & Absent \\
         \hline
    \end{tabular}
    \caption{Comparison of twists in color codes.  ${}^\dagger$One can also obtain a CSS code} by choosing $Z$ or $X$ type stabilizer on twist face.
    \label{tab:comp_twists}
\end{table}

A comparison of various twist types is given in Table~\ref{tab:comp_twists}.
For completeness we have also included all the three fundamental types of twists in color codes. For the rest of the paper we only consider charge and color permuting twists. 
In the following sections we assume that the lattices are obtained by embedding graphs on a two dimensional plane.
We also assume that the boundary of lattices have the same color as given in Fig.~\ref{fig:sq-octagon-lattice} and Fig.~\ref{fig:hexagon-lattice-boundary}.

\subsection{Pauli operators as strings in color codes without twists}
In this section, we give a mapping between Pauli operators and strings on color code lattices without twists.
This correspondence is useful in abstracting away the  lattice information and studying the code properties using the string algebra.
The three single qubit Pauli operators, namely, $X$, $Y$ and $Z$ are represented using three different types of strings.
A solid string of any color represents Pauli operator $Z$, a dashed string of any color represents Pauli operator $X$ and a dash-dotted string of any color represents Pauli operator $Y$, see Fig.~\ref{fig:cc-pauli-string-map}.

\begin{figure}[htb]
    \centering
    \includegraphics[scale = 1.5]{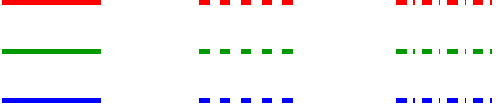}
    \caption{String to Pauli operator correspondence: solid string of any color is mapped to $Z$, dashed string of any color is mapped to $X$ and dash dotted string of any color is mapped to $Y$.}
    \label{fig:cc-pauli-string-map}
\end{figure}

We begin by representing a Pauli operator on a vertex in a color code without twists.
When a Pauli error occurs on a vertex, syndromes are created on all the three faces incident on that vertex, see Fig.~\ref{fig:Qb-cc-pauli-string-1}.
This is because, all the faces have both $Z$ and $X$ type stabilizers defined on them and an error violated at least one of these stabilizers.
A single Pauli is represented by a string  with three terminals as shown in  Fig.~\ref{fig:Qb-cc-pauli-string-5}.
The end point of the strings are interpreted as syndromes.
Now consider applying the same error operator on two adjacent vertices as shown in Fig.~\ref{fig:Qb-cc-pauli-string-3}.
This can be understood from stabilizer viewpoint.
The error operator commutes with the stabilizer of the blue face and hence no syndrome should be observed on this face.
Similarly, on the green face, the syndromes vanish.
(In effect two $bz$ syndromes fuse to vacuum.)
The string representation of Fig.~\ref{fig:Qb-cc-pauli-string-3} is shown in Fig.~\ref{fig:Qb-cc-pauli-string-7}.
This representation is obtained from the single qubit representation shown in Fig.~\ref{fig:Qb-cc-pauli-string-5}.
The two stings on blue and green faces are connected as these faces do not host syndromes.
On the other hand, on the red faces the string is terminated indicating the presence of syndromes.

Note that if one of the red faces hosts a syndrome $rz$, then applying the operator shown in Fig.~\ref{fig:Qb-cc-pauli-string-4} annihilates the syndrome on the face with prior syndrome and creates one on other without any prior syndrome.
This process can be seen as moving the syndrome between adjacent faces of the same color.
The operators that move syndromes are known as hopping operators~\cite{BhagojiSarvepalli2015}.
A compact representation of the string in Fig.~\ref{fig:Qb-cc-pauli-string-7} is given in Fig.~\ref{fig:Qb-cc-pauli-string-8}.
Note that the string is drawn dashed indicating that is represents Pauli operator $X$ and its color is red indicating that red edges are in its support.

\begin{figure}[htb]
    \centering
    \begin{subfigure}{.225\textwidth}
        \centering
        \includegraphics[scale = 1]{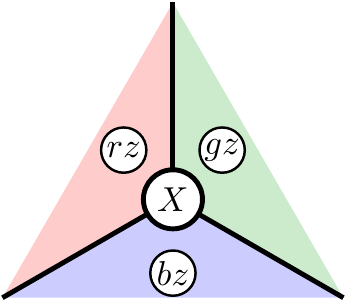}
        \subcaption{}
        \label{fig:Qb-cc-pauli-string-1}
    \end{subfigure}
     ~
    \begin{subfigure}{.225\textwidth}
        \centering
        \includegraphics[scale = 1]{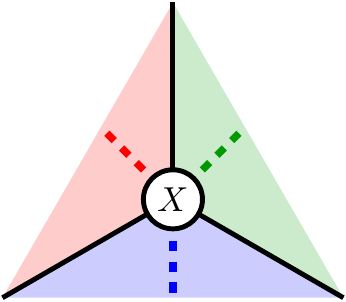}
        \subcaption{}
        \label{fig:Qb-cc-pauli-string-5}
    \end{subfigure}
     ~
    \begin{subfigure}{.225\textwidth}
        \centering
        \includegraphics[scale = 1.15]{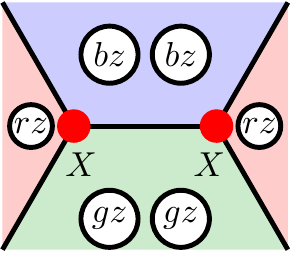}
        \subcaption{}
        \label{fig:Qb-cc-pauli-string-3}
    \end{subfigure}
    ~
    \begin{subfigure}{.225\textwidth}
        \centering
        \includegraphics[scale = 1.15]{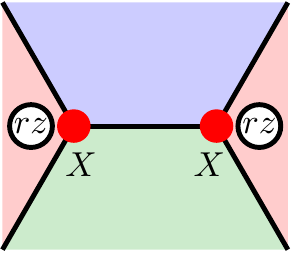}
        \subcaption{}
        \label{fig:Qb-cc-pauli-string-4}
    \end{subfigure}
     ~
    \begin{subfigure}{.225\textwidth}
        \centering
        \includegraphics[scale = 1.15]{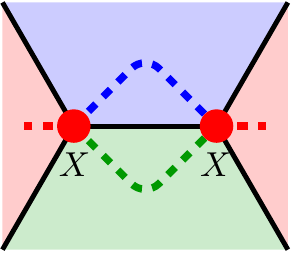}
        \subcaption{}
        \label{fig:Qb-cc-pauli-string-7}
    \end{subfigure}
     ~
    \begin{subfigure}{.225\textwidth}
        \centering
        \includegraphics[scale = 1.15]{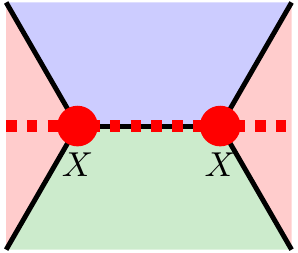}
        \subcaption{}
        \label{fig:Qb-cc-pauli-string-8}
    \end{subfigure}
\caption{Mapping Pauli operators to strings. (a) An $X$ error on the vertex violates $Z$ type stabilizers on the three faces creating three nonzero syndromes. (b) String representation of the single qubit $X$ error. The end points of strings indicate syndromes. (c) Syndromes of a two qubit $X$ error on adjacent qubits. This creates a syndrome on red faces and two syndromes on green and blue faces. (d) Syndromes on blue and green faces are annihilated since the two qubit $X$ operator commutes with stabilizers of those faces. (e) String representation for the two qubit error in Fig.~\ref{fig:Qb-cc-pauli-string-3}. String terminates in faces with syndrome and is continuous in faces without syndrome. (f) Alternative representation for the error in Fig.~\ref{fig:Qb-cc-pauli-string-3}. This string is obtained from Fig.~\ref{fig:Qb-cc-pauli-string-7} by combining blue and green strings.}
\label{fig:str_pauli}
\end{figure}

\section{Charge permuting twists}
\label{sec:charge}
In this section, we discuss charge permuting twists.
The process of twist creation and movement corresponds to code deformation~\cite{Bombin2009}.
After reviewing the procedure to create charge permuting twists,
we derive the number of encoded logical qubits by stabilizer counting.
We also give logical operators for the encoded qubits.
We give a mapping between Pauli operators and strings on the lattice.

Charge permuting twists are faces in the lattice that permute the Pauli labels of anyons that encircle them.
A $p$ type charge permuting twist acts as follows:
\begin{subequations}
\begin{eqnarray}
  cp &\mapsto& cp \\
  cp^\prime &\mapsto& cp^{\prime \prime}\\
  cp^{\prime \prime} &\mapsto& cp^\prime
\end{eqnarray}
\label{eqn:p-twist}
\end{subequations}
where $p$, $p^\prime$ and $p^{\prime \prime}$ are distinct Pauli labels.
An example is the $Y$ twist~\cite{Kesselring2018} that exchanges the $X$ and $Z$ labels of anyons while leaving the $Y$ label unchanged.
The $X$ twist leaves the $X$ label unchanged and exchanges the $Y$ and $Z$ labels, see Fig.~\ref{fig:charge-permuting-twists-action}, and the $Z$ twist exchanges $X$ and $Y$ labels leaving the $Z$ label unchanged.
Lattice modification is not required for creating charge permuting twists but only stabilizers of certain faces need to be modified~\cite{Kesselring2018}.

Before moving on to twist creation, we need a notion of how far apart the twists are as the code distance depends on it.

The dual of a lattice $\Gamma(V, E,F)$ is constructed by mapping the faces in $\Gamma$ to vertices, and connecting two such vertices if the  faces in the primal lattice share an edge. 
This construction maps faces to vertices, edges to edges and vertices to faces in the dual lattice. 
Now we can precisely state how  far apart two charge permuting twists are.

\begin{definition}[Twist separation]
Separation between twists is the distance between the vertices corresponding to twists in the dual lattice.
\label{def:twist-separation}
\end{definition}

As an example, consider the separation between twists in  Fig.~\ref{fig:charge-permuting-twists} where the separation between twists $t_1$, $t_2$ and $t_3$, $t_4$ is three and four respectively.
The domain wall in Fig.~\ref{fig:charge-permuting-twists} between $t_3$ and $t_4$ has been drawn to illustrate that a domain wall need not be the shortest path between twist faces.
The twist separation in that case still conforms to the above definition.
Twist separation determines the code distance, for this reason we include this as a design parameter while constructing color codes with twists.

\subsection{Creation of charge permuting twists}
We now summarize the procedure to create $k$ pairs of charge permuting twists in a $2$-colex. 
In case of charge permuting twists, there are no physical dislocations in the lattice. 
Twists are introduced by breaking the uniformity of  stabilizer assignment for each face of the 2-colex.
To introduce $k$ pairs of charge permuting twists, we first need to choose $2k$ vertices in the dual lattice such that any pair of them are at least a distance $\ell \ge 2$ apart.
Now, connect a disjoint pair of vertices by a path so that no two paths crossover.
The faces corresponding to the terminal points of a path become the twists while the path connecting them becomes a domain wall. 
The faces which correspond to the intermediate vertices of the path are the ones through which the domain wall passes through, see Fig.~\ref{fig:charge-permuting-twists}.
Once the twists and the faces through which domain wall passes through are chosen, assign the stabilizers on them as discussed next.

\begin{figure}[htb]
    \centering
    \begin{subfigure}{.45\textwidth}
        \centering
        \includegraphics[scale = 1.15]{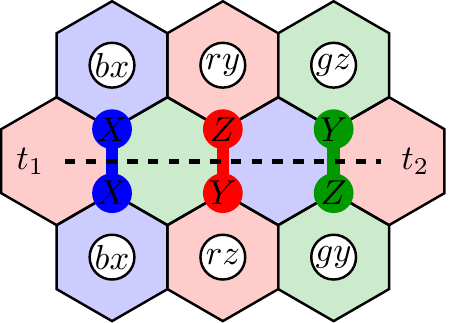}
        \subcaption{}
    \label{fig:charge-permuting-twists-action}
    \end{subfigure}
    ~
    \begin{subfigure}{.45\textwidth}
        \centering
        \includegraphics[scale = .9]{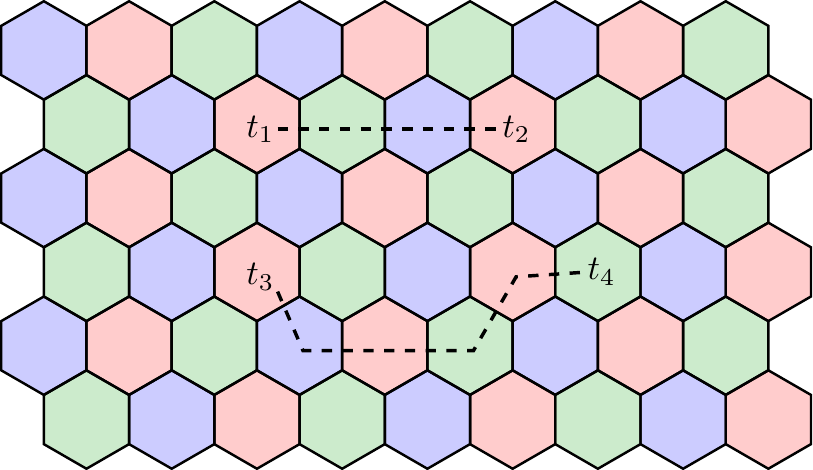}
        \subcaption{}
    \label{fig:charge-permuting-twists}
    \end{subfigure}
    \caption{Charge permuting twists and their action on the anyons. (a) Action of $X$ charge permuting twists on anyons. These twists permute only the charge of the anyons. An anyon with Pauli label $X$ is left unchanged whereas the anyons with $Z$ and $Y$ label are exchanged as shown. (b) Charge permuting twists. In the figure, faces marked $t_1$, $t_2$, $t_3$ and $t_4$ are twists and the dashed line is the domain wall. Charge permuting twists need not be of the same color, see twist pair $t_3$ and $t_4$ which are faces of different color. The domain wall need not be the shortest path between the twist faces, see twist pair $t_3$ and $t_4$.}
\end{figure}

Faces in color code lattices with charge permuting twists can be grouped into three categories: i) twists, ii) faces through which domain wall passes through and iii) faces that are neither twists nor faces through which domain wall passes through.
Stabilizers for the $Y$ twist are given in Ref.~\cite{Kesselring2018}.
We give the stabilizer assignment for the $X$ twist.

Two stabilizers are defined on faces that are neither twists nor through which domain wall passes through:
\begin{equation}
B_f^Z = \prod_{v \in V(f)} Z_v, \text{ and  } B_f^Y = \prod_{v \in V(f)} Y_v. \label{eqn:stab-unmodified-faces}
\end{equation}

The faces through which domain wall passes through has to permute the labels $Y$ and $Z$ of anyons while leaving the $X$ label unchanged.
This will happen if the Pauli operators of stabilizers on one side of the domain wall are of $Y$ type ($Z$ type) and on the other side they are of $Z$ type ($Y$ type).
Let the domain wall partition the vertices of a face $f$ into two sets $M_1$ and $M_2$.
Two stabilizers are assigned to such faces as defined below:
\begin{subequations}
\begin{eqnarray}
  B_{f,1} &=& \prod_{v \in M_1} Y_v \prod_{v \in M_2} Z_v\\
  B_{f,2} &=& \prod_{v \in M_1} Z_v \prod_{v \in M_2} Y_v.
\end{eqnarray}
\label{eqn:charge-domain-stab}
\end{subequations}

\begin{figure}[htb]
\centering
\includegraphics[scale = .9]{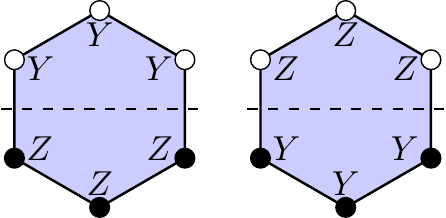}
\caption{Vertex partition for faces comprising domain walls in charge permuting twists. 
Dashed line is the domain wall and it does not pass through any vertex. 
The vertices indicated as dark and white circles form the partition of the vertices of the face.}
\label{fig:domain-wall-stab}
\end{figure}

The stabilizers defined in Equations~\eqref{eqn:charge-domain-stab} are shown in Fig.~\ref{fig:domain-wall-stab}.
The stabilizer assigned to the twist face has to commute with the stabilizer defined on the face adjacent to it through which domain wall passes through.
These two faces share exactly one edge.
Given the assignment in Eq.~\eqref{eqn:charge-domain-stab}~and~\eqref{eqn:stab-unmodified-faces},  
the only possible stabilizer that can be assigned to the twist face is that of $X$ type,
\begin{equation}
    B_{\tau} = \prod_{v \in V(\tau)} X_v.
    \label{eqn:charge-twist-stab}
\end{equation}

\noindent \emph{Stabilizer dependency.} The stabilizer assignment in Eq.~\eqref{eqn:stab-unmodified-faces}--\eqref{eqn:charge-twist-stab} is constrained by the following relations:
\begin{subequations}
\begin{eqnarray}
\prod_{f \in T\cap \mathsf{F}_{rb}}  B_f \prod_{f \in \mathsf{F}_{rb} \cap M}  B_{f,1}B_{f,2} \prod_{f \in \mathsf{F}_{rb} \cap U} B_{f}^ZB_{f}^Y&= &I\\
\prod_{f \in T\cap \mathsf{F}_{rg}}  B_f \prod_{f \in \mathsf{F}_{rg} \cap M}  B_{f,1}B_{f,2} \prod_{f \in \mathsf{F}_{rg} \cap U} B_{f}^ZB_{f}^Y&= &I
\end{eqnarray}
\label{eqn:charge-perm-stab-constraint}
\end{subequations}
where $T$, $M$ and $U$ denote the set of twist faces, modified faces and unmodified faces respectively.
From Equations~\eqref{eqn:charge-perm-stab-constraint}, we can see that one of the green and blue face stabilizers is dependent.
We take the outer unbounded blue face stabilizer to be the dependent one.
Since two stabilizers are defined on the unbounded blue face and only one of them is dependent, we have to measure the other non-local stabilizer which is undesirable.
We therefore discard one of the stabilizers and redefine the unbounded blue face stabilizer as 
\begin{equation}
    B_{f_e} = \prod_{v \in V(f_e)} X_v.
\end{equation}
This stabilizer still satisfies the constraint in Equations~\eqref{eqn:charge-perm-stab-constraint}.

\noindent \emph{Stabilizer commutation.}
Note that any two adjacent faces share two common vertices.
The Pauli operators corresponding to stabilizers are either same or different on the common vertices.
Hence, stabilizer commutation follows.
A more detailed analysis is given in Appendix~\ref{sec:charge-stabilizer}.

\begin{remark}
The stabilizers of color code with $X$ twist can be obtained from that of the $Y$ twist by applying phase gate on all qubits.
\end{remark}

We summarize with the following result on the parameters of color codes with 
charge permuting twists. 
\begin{lemma}[Encoded qubits for charge permuting twists]
A color code lattice with $t$ charge permuting twists encodes 
$\left(t - 1\right)$ logical qubits.
\label{lm:encoded-qubits-charge}
\end{lemma}

\begin{proof}
In a color code lattice with charge permuting twists, all vertices are trivalent.
Therefore, we have $2e = 3v$, where $e$ and $v$ are the number of edges and vertices of the lattice respectively.
Using this in the Euler's formula, $v + f - e = 2$, where $f$ is the number of faces, we get, $2f = v + 4$.
Only one stabilizer is defined on twist faces and the unbounded face while the remaining faces we define two stabilizer generators.
Therefore, the number of stabilizers is $ 2f - t - 1 = v + 3 - t$.
Accounting for the dependencies (see Equations~\eqref{eqn:charge-perm-stab-constraint}), we obtain, 
$s = v - t + 1 = v - 2 (t /2 - 1) - 1$ independent stabilizer generators.
Since total number of qubits is $v$, the number of encoded qubits is $k = t - 1 = 2 (t /2 - 1) + 1$.
\end{proof}

\subsection{Pauli operators as strings on a  color code lattice with charge permuting twists} 
The  representation of the Pauli operators by strings as given in Fig.~\ref{fig:str_pauli} does not work in the presence of charge permuting twists as these twists permute the charge label of anyons.
Consider the operator shown in Fig.~\ref{fig:Qb-cc-pauli-string-6}.
The upper blue face is twist, domain wall passes through the green face and red faces are neither twists nor faces through which domain wall passes through.
Red faces are normal faces.
The Pauli operator $Z_u Y_v$ creates two $bx$ syndromes on twist face and syndromes $gy$ and $gz$ on the green face.
Also, $ry$ and $rz$ syndromes are created on red faces.
Syndromes on the twist face annihilate each other and hence the string is continuous in twist face.

However, the Pauli operator $Z_u Y_v$  commute with the stabilizers of the green face.
Therefore, string is continuous in green face too, see Fig.~\ref{fig:Qb-cc-pauli-string-2}.
The red faces host anyons and hence the string terminates there.
An alternate representation for  the operator described in Fig.~\ref{fig:Qb-cc-pauli-string-6} is given in Fig.~\ref{fig:Qb-cc-pauli-string-9}.
Note that in Fig.~\ref{fig:Qb-cc-pauli-string-9}, the string changes midway as it crosses the domain wall.

\noindent 
{\em String representation of stabilizers.}
We represent the stabilizers using the string representation just developed.
Twist stabilizer in string notation is given in Fig.~\ref{fig:Qb-cc-charge-twist-stab-string}.
Note that the twist is assigned $X$ type stabilizer and hence the string is dashed.
Also, the $X$ twist does not permute $X$ label and hence the string is not changed as it crosses the domain wall.
An alternate representation of twist stabilizer in Fig.~\ref{fig:Qb-cc-charge-twist-stab-string} would be to take green string instead of blue.
Stabilizers of faces through which domain wall passes through is given in Fig.~\ref{fig:Qb-cc-charge-modified-stab-string-1} and Fig.~\ref{fig:Qb-cc-charge-modified-stab-string-2}.
These stabilizers are of mixed type with $Z$ ($Y$) on vertices on one side of the domain wall and $Y$ ($Z$) on vertices on the other side of domain wall.
Note that the strings change as they cross domain wall.
For other faces that are not twists and through which domain wall does not pass through, the strings are of $Z$ type and $Y$ type.

\begin{figure}[htb]
    \centering
    \begin{subfigure}{.225\textwidth}
        \centering
        \includegraphics[scale = 1]{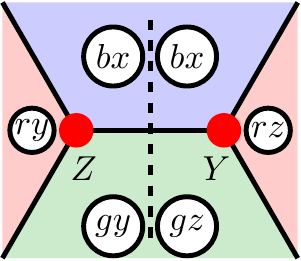}
        \subcaption{}
        \label{fig:Qb-cc-pauli-string-6}
    \end{subfigure}
    ~
    \begin{subfigure}{.225\textwidth}
        \centering
        \includegraphics[scale = 1]{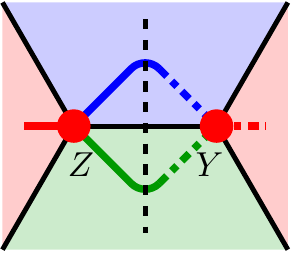}
        \subcaption{}
        \label{fig:Qb-cc-pauli-string-2}
    \end{subfigure}
     ~
    \begin{subfigure}{.225\textwidth}
        \centering
        \includegraphics[scale = 1]{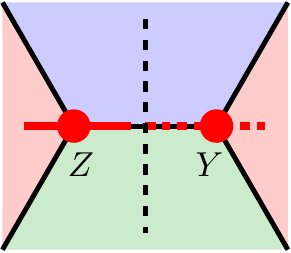}
        \subcaption{}
        \label{fig:Qb-cc-pauli-string-9}
    \end{subfigure}
    ~
    \begin{subfigure}{.225\textwidth}
        \centering
        \includegraphics[scale = 1.15]{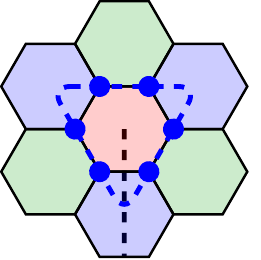}
        \subcaption{}
        \label{fig:Qb-cc-charge-twist-stab-string}
    \end{subfigure}
    ~
    \begin{subfigure}{.225\textwidth}
        \centering
        \includegraphics[scale = 1.15]{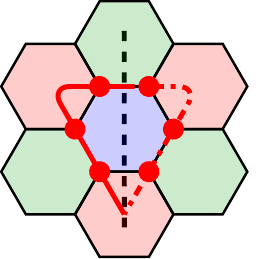}
        \subcaption{}
        \label{fig:Qb-cc-charge-modified-stab-string-1}
    \end{subfigure}
     ~
    \begin{subfigure}{.225\textwidth}
        \centering
        \includegraphics[scale = 1.15]{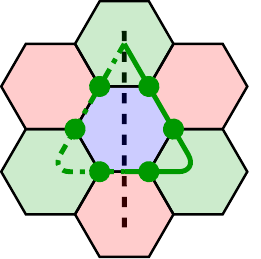}
        \subcaption{}        
        \label{fig:Qb-cc-charge-modified-stab-string-2}
    \end{subfigure}
    \caption{Strings and stabilizers in the presence of charge permuting twists. (a) Syndromes introduced on  blue twist face and  green face. The syndromes on twist face and green face vanish as the error operator commutes with the stabilizers. (b) Strings are continuous in blue and green faces indicating no syndromes are present on them. Open string in red faces indicate syndrome on them. (c) String representation of Fig.~\ref{fig:Qb-cc-pauli-string-6}. Note that syndrome changes its charge while retaining color. (d) String representation of twist stabilizer. The string does not change charge as the twist does not permute $X$ label. (e) String representation of a stabilizer on the face through which domain wall passes through. (f) String representation of another stabilizer on face through domain wall passes through.}
    \label{fig:Qb-cc-charge-modified-stab-string}
\end{figure}

\noindent 
\emph{String representation of logical operators.}
Logical operators (including stabilizers) are closed strings encircling an even number of twists.
A closed string enclosing exactly a pair of twists is a nontrivial logical operator. 
Further, we also need to show that these operators are not generated by stabilizers.
This is done by showing the existence of a string whose operator anticommutes with that of the given string, see Fig.~\ref{fig:LO-X-charge-permuting-lattice} and Fig.~\ref{fig:LO-Z-charge-permuting-lattice}.
These operators are logical $X$ and $Z$ operators respectively.
An important property of a logical operator is that it commutes with all the stabilizers.
Every closed string must enter and exit a face an even number of times. 
Therefore, every face shares an even number of common vertices with the logical operator.
This also holds for faces through which domain wall passes through.
Taking into account the stabilizer structure of domain wall faces, we conclude that the operators shown in Fig.~\ref{fig:LO-charge-permuting-lattice} commute with all stabilizer generators.

\begin{figure}[htb]
    \centering
    \begin{subfigure}{.225\textwidth}
          \centering
          \includegraphics[scale = .9]{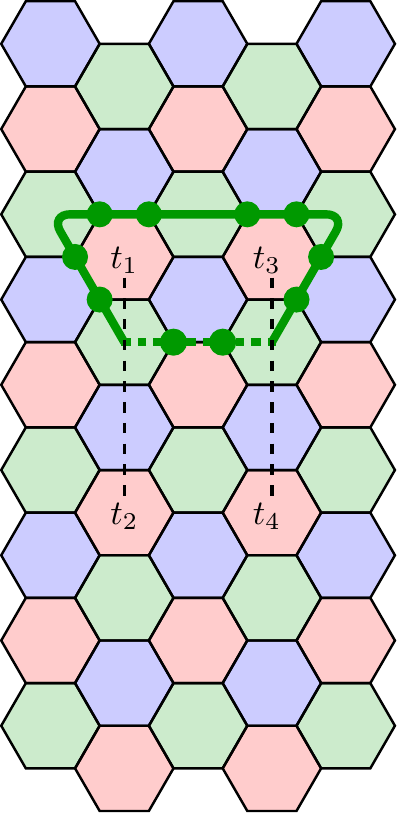}
          \subcaption{}
          \label{fig:LO-X-charge-permuting-lattice}
    \end{subfigure}
    ~
    \begin{subfigure}{.225\textwidth}
          \centering
          \includegraphics[scale = .9]{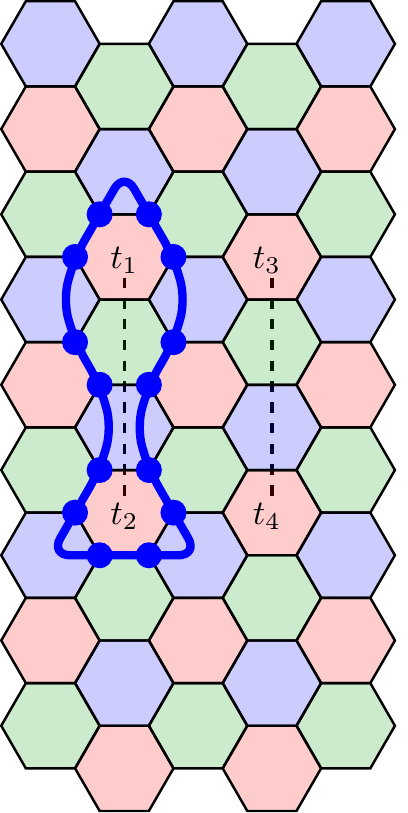}
          \subcaption{}
          \label{fig:LO-Z-charge-permuting-lattice}
    \end{subfigure}
    \caption{Logical operators for one encoded qubit in color code lattice with $X$ type charge permuting twists. Note that both operators have one vertex in common where the Pauli operator differs. (a) Logical $X$ operator depicted on lattice. (b) Logical $Z$ operator depicted on lattice.}
    \label{fig:LO-charge-permuting-lattice}
\end{figure}

Note that the product of twist stabilizers $t_1$, $t_2$ and the stabilizers defined on faces through which domain wall passes through (except the blue face) has the same support as the operator shown in Fig.~\ref{fig:LO-Z-charge-permuting-lattice} with $X$ on vertices in the support.
The product of these stabilizers with the logical $Z$ operator in Fig.~\ref{fig:LO-Z-charge-permuting-lattice} flips the Pauli operator on all vertices in its support to $Y$.
Therefore, the two operators are equivalent up to stabilizers.
The same holds true for logical $X$ operator shown in Fig.~\ref{fig:LO-X-charge-permuting-lattice}.
Pauli operators $Y$ and $Z$ can be flipped by taking the product of twist stabilizers on $t_1$ and $t_3$ and the blue face on which the operator has support.

\begin{figure}[htb]
        \centering
        \begin{subfigure}{.45\textwidth}
            \centering
            \includegraphics[scale = .85]{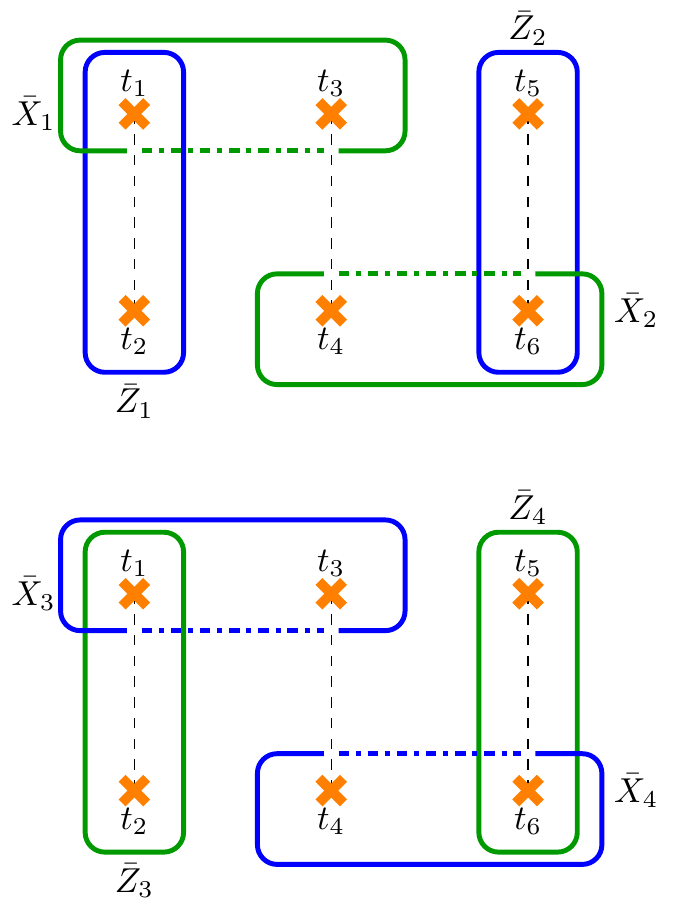}
            \subcaption{}
        \label{fig:lo_charge}
        \end{subfigure}
        ~
        \begin{subfigure}{.45\textwidth}
            \centering
            \includegraphics[scale = .75]{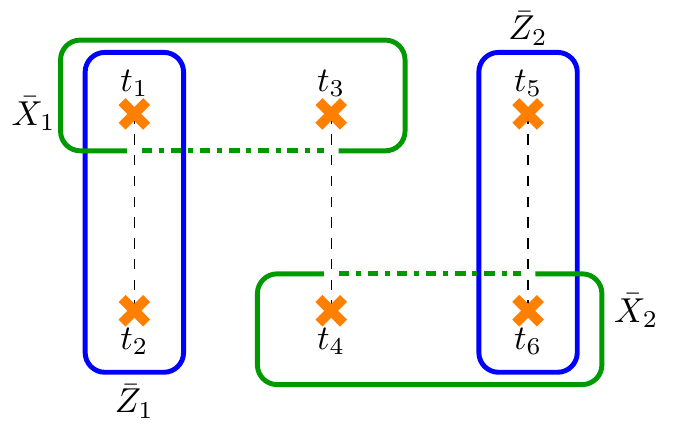}
            \subcaption{}
        \label{fig:lo_charge_canonical}
        \end{subfigure}
        ~
        \begin{subfigure}{.45\textwidth}
            \centering
            \includegraphics[scale = .75]{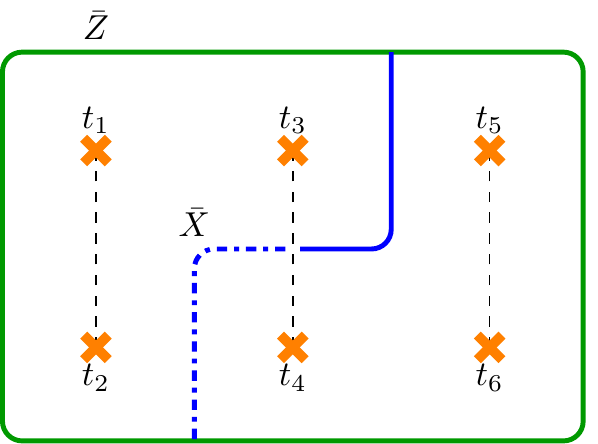}
            \subcaption{}
            \label{fig:LO-charge-perm-extra-qubit}
\end{subfigure}
\caption{Logical operators for encoded qubits in lattices with $X$ type charge permuting twists. (a) Logical operators for encoded qubits in lattice having six $X$ type charge permuting twists. (b) Canonical logical operators for encoded qubits in lattice having six $X$ type charge permuting twists for $\lfloor t/3 \rfloor$ encoding. (c) Logical operators for the the logical qubit in the presence of a pair of twists. The green string has support on the physical qubits on the boundary.}
\label{fig:charge_LO}
\end{figure}

For purposes of simple representation of twist and logical operators, we hide the lattice information and represent only twists, domain walls and logical operators.
Logical operators are represented as closed strings.
These strings are edges in a modified lattice  of the color code.
These modified lattices are  called shrunk lattices. 
A shrunk lattice of color $c$ is obtained by shrinking all faces of color $c$ in the lattice to vertices.
Equivalently, we can contract the edges of the lattice that are not colored $c$.
Observe that upon crossing the domain wall, color of the string is unaffected.
Logical operators when three pairs of $X$ twists are present in the lattice is given in Fig.~\ref{fig:lo_charge}.

The disadvantage with this choice of logical operators is that whenever a gate is to be performed only on a particular encoded qubit by braiding, say encoded qubit $1$ ($2$), we also end up performing a gate on logical qubit $3$ ($4$) (see Fig.~\ref{fig:lo_charge}) which is undesirable.
Therefore, we treat encoded qubits $3$ and $4$ as gauge qubits.
The encoded qubits that do not carry any useful information are called \textit{gauge qubits}.

The canonical form of logical operators for $\lfloor t/3 \rfloor$ encoding is shown in Fig.~\ref{fig:lo_charge_canonical}.
With $\lfloor t/3 \rfloor$ encoding, a closed string encircling three pairs of twists encoding two logical qubits is a stabilizer.
We use $\lfloor t / 3 \rfloor$ encoding which encodes two qubits per six twists for implementing gates by braiding.

\begin{theorem}
A color code with $t$ charge permuting twists with $ \left\lfloor \frac{t}{3} \right\rfloor$ encoding defines an $\left[ \left[n, \left\lfloor \frac{t}{3} \right\rfloor\right]\right]$
subsystem code 
with $\left(t - 1\right) - \left\lfloor \frac{t}{3} \right\rfloor $  gauge qubits.
\end{theorem}

\begin{remark}
By adding the gauge operators to the canonical logical operators, equivalent logical operators can be obtained.
For instance, logical operators $\bar{Z}_1$ and $\bar{X}_1$ are equivalent to the operators corresponding to red strings in Fig.~\ref{fig:Qb-cc-logY-1} and Fig.~\ref{fig:Qb-cc-logY-2} respectively.
\end{remark}

\noindent \emph{Code distance}.
Distance of the code is the minimum weight of the operator in $C(S) \setminus  S$ where $S$ is the stabilizer and $C(S)$ is the centralizer of stabilizer. 
The centralizer of a subgroup $S$ of a group $G$ is the set of all elements of $G$ that commute with every element of the group. It is denoted as denoted as $C(S)$ and formally defined as 
\begin{equation*}
    C(S) = \{g \in G \text{ } \vert \text{ } sg = gs, \text{ for all }  s \in S \}.
\end{equation*}
The centralizer and normalizer are same for subgroups of the Pauli group.
A logical operator corresponds to a nontrivial closed string that encircles an even number of twists.
As in the case of surface codes with twists, the shortest of the nontrivial strings are those that encircle a pair of twists.
These strings have weight $O(\ell)$, where $\ell$ is the separation between twists.
Therefore, we conjecture that the distance of the code is $O(\ell)$.

\begin{figure}[htb]
    \centering
    \begin{subfigure}{.45\textwidth}
        \centering
        \includegraphics[scale = .8]{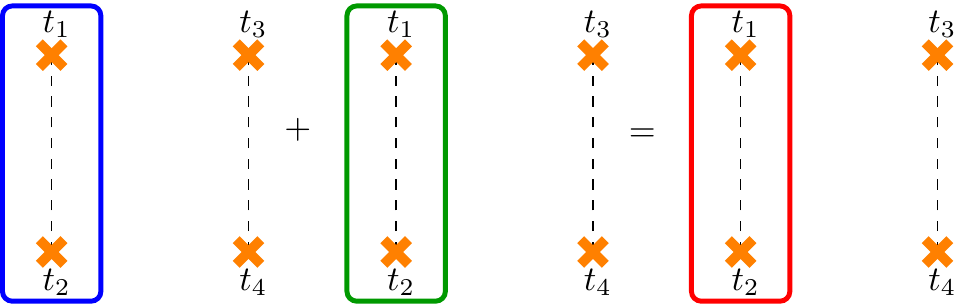}
        \subcaption{}
        \label{fig:Qb-cc-logY-1}
    \end{subfigure}
    ~
    \begin{subfigure}{.45\textwidth}
        \centering
        \includegraphics[scale = .8]{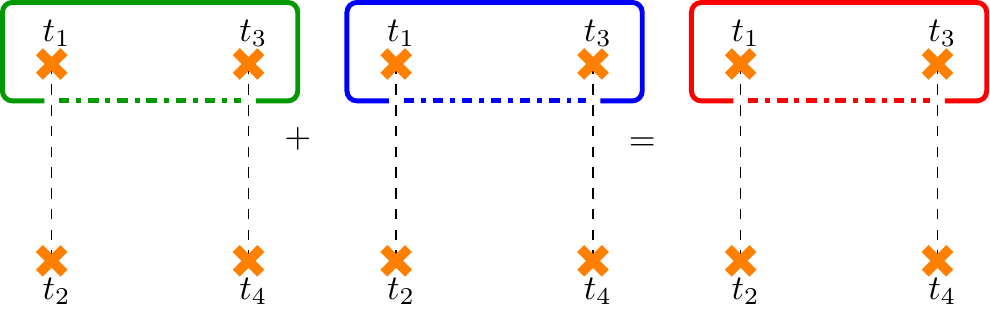}
        \subcaption{}
        \label{fig:Qb-cc-logY-2}
    \end{subfigure}
    \caption{The logical operators for $t/3$-encoding  can also be represented as closed strings of different color by adding gauge operators. Equivalent representation of $\bar{Z}_1$ and $\bar{X}_1$ are shown. (a) Adding a gauge operator (green string) to $\bar{Z}_1$ operator. The result is a red string encircling twists $t_1$ and $t_2$. (b) Logical $\bar{X}_1$ operator, shown as green string, can also be represented an equivalent red string by adding a gauge operator (blue string).}
    \label{fig:equivalent_LO}
\end{figure}

We use the charge permuting twists constructed in this section for realizing encoded Clifford gates by braiding. 
These results are presented in Section~\ref{sec:gates-charge-permuting}.

\section{Color permuting twists }
\label{sec:color}
In this section, we present the construction of color permuting twists.
Our construction differs from the one given in Ref.~\cite{Kesselring2018} and is more along the lines of Ref.~\cite{Bombin2011}.
We first address the modification of an arbitrary $2$-colex to create color permuting twists.
While creating twists, we should also ensure that the code resulting from lattice modification has the specified distance.
This is accomplished by tailoring the separation between twists.
We then discuss stabilizer assignment, code parameters and logical operators.

Color permuting twists are faces in the lattice which permute the color label of anyons that encircle them.
A $c$-type color permuting twist acts as follows:
\begin{subequations}
\begin{eqnarray}
  cp &\mapsto& cp\\
  c^\prime p&\mapsto& c^{\prime \prime} p\\
  c^{\prime \prime} p &\mapsto& c^\prime p
\end{eqnarray}
\label{eqn:c-twist}
\end{subequations}
where the colors  $c$, $c^\prime$ and $ c^{\prime \prime}$ are all distinct.
Observe that color permuting twists fix one color and exchange the other two.
The Pauli label of an anyon is not altered by color permuting twists.

\subsection{Creation and movement}

\noindent \emph{Intuition.} In a $2$-colex, edges of a specific color connect the faces of the same color, see Fig.~\ref{fig:int-color-perm-0}.
As a result, the syndromes can only be moved between faces of same color.
In order to create color permuting twist, we need to introduce an edge between two faces of different color.
Suppose that we introduce an edge $(u,v)$ between blue and green face, see Fig.~\ref{fig:int-color-perm-1}.
This edge cuts through the common red face to the blue and green faces.
This move also makes the vertices $u$ and $v$ tetravalent and splits the red face into two: the upper pentagon face and the lower triangular face.
It is not possible to assign stabilizers to the new red faces such that they commute with the rest~\cite{Anderson2013}.
Therefore, the vertices $u$ and $v$ have to be made trivalent.
This is done by removing the common neighbor $w_1$ to vertices $u$ and $v$, see Fig~\ref{fig:int-color-perm-2}.
The vertex $w_2$ is rendered two-valent by this move. 
The neighbors of the vertex $w_2$ are connected by an edge making them tetravalent.
The tetravalency of the neighbors of $w_2$ is resolved by removing the vertex $w_2$, see Fig~\ref{fig:int-color-perm-3}.
Note that the new edge also connects blue and green faces.
The face in between twists has mixed color.
One can assign either blue or green color to this face.
Similar procedure is followed for moving twists apart.
The twist created in this illustration is red twist since it permutes green and blue labels leaving the red label unchanged.
Other twists like green twist and blue twist are defined similarly.

\begin{figure}[htb]
\centering
    \begin{subfigure}{.225\textwidth}
        \centering
        \includegraphics[scale = 1]{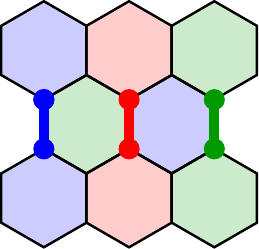}
        \subcaption{}
        \label{fig:int-color-perm-0}
    \end{subfigure}
    ~
    \begin{subfigure}{.225\textwidth}
        \centering
        \includegraphics[scale = 1]{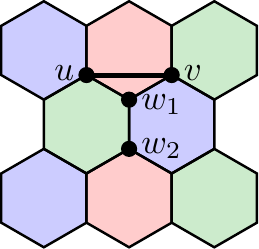}
        \subcaption{}
        \label{fig:int-color-perm-1}
    \end{subfigure}
    ~
    \begin{subfigure}{.225\textwidth}
        \centering
        \includegraphics[scale = 1]{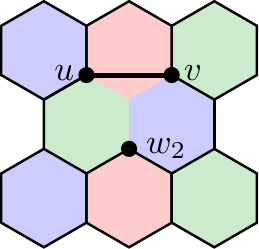}
        \subcaption{}
        \label{fig:int-color-perm-2}
    \end{subfigure}
    ~
    \begin{subfigure}{.225\textwidth}
        \centering
        \includegraphics[scale = 1]{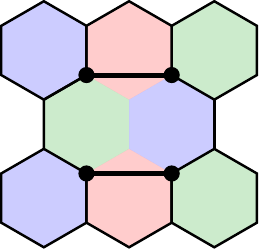}
        \subcaption{}
        \label{fig:int-color-perm-3}
    \end{subfigure}
    \caption{Process of twist creation in a $2$-colex. (a) In a $2$-colex, edges connect faces of the same color. Red, green and blue edges connect the faces of respective color as shown. (b) To create twist, introduce an edge between two faces of different color. This renders two vertices tetravalent. (c) The vertex $w_1$ is removed to resolve the tetravalency of vertices $u$ and $v$. This move also makes the vertex $w_2$ two-valent. (d) The neighbors of two-valent vertex $w_2$ are connected by an edge making them tetravalent. The tetravalency is resolved by removing the vertex $w_2$.}
\end{figure}

\begin{remark}
Note that the twist faces have an odd number of edges.
This breaks the three-face-colorability in the vicinity of such faces of the lattice and creates color permuting twists.
\end{remark}

We would like to introduce $k \ge 1$ pairs of twists of the same color $c$ in a $2$-colex. 
To do so, we first choose an even number of vertices of the same color $c$ in the dual lattice such that any pair of them is at least a distance $\ell$ apart.
A path is introduced between disjoint pairs of vertices so that no two paths overlap.
The challenge faced in constructing color permuting twist is mainly the selection of faces in such a way that the intermediate faces between twists do not grow unbounded.
The aforesaid path is not arbitrary but has to satisfy the properties given below.
\begin{compactenum}[C1)]
    \item The faces corresponding to the nearest $c$-colored vertices on the path are connected by an edge in the original lattice.
    \item Let $u$, $v$ and $w$ are three consecutive $c$-colored vertices on the path and such that distance between $u$ and $v$ is two and distance between $v$ and $w$ is two in the dual lattice.
    The vertex $w$ cannot be at a distance of two from $u$.
\end{compactenum}
The first condition ensures that the faces with color $c^\prime$ and $c^{\prime \prime}$ are merged during twist creation.
The second condition ensures that the intermediate faces between twists do not grow in size.
With this, one can create an even number of twists in an arbitrary $2$-colex.
Note that we do not use edge coloring in our work.
The procedure to create color permuting twists in an arbitrary $2$-colex is given in Algorithm \ref{alg:color-permuting-twists}.

\begin{algorithm}
\caption{Algorithm to create $k$ pairs of $c$-color permuting twists in a $2$-colex.}
    \label{alg:color-permuting-twists}
    \begin{flushleft}
    \algorithmicrequire{ $2$-colex, separation between twists $\ell \ge 2$.}\\
    \algorithmicensure{ Trivalent lattice with $k$ pairs of color permuting twists.}
    \end{flushleft}
    
\begin{algorithmic}[1]
        \State Choose $2k$ vertices of color $c$ in the dual lattice such that any pair of them is at least a distance $\ell$ apart.
        Introduce a path between disjoint pairs of vertices so that no two paths overlap.
        The $c$-colored vertices on a given path $\pi_i$ satisfy conditions C1) and C2).
        Let $\{ e_j^{(i)} \}_{j = 1}^{\ell}$ be the collection of such edges corresponding to the $i^{th}$ path $\pi_i$ in the original lattice, see Fig.~\ref{fig:tc0}.
    	\For{ $i = 1$ to $k$}
    	\State Remove the vertices incident on the edge $e_1^{(i)}$ (see Fig.~\ref{fig:tc1}) and connect the closest two-valent vertices, see Figs.~\ref{fig:tc2},~\ref{fig:tc3}. \Comment {Twist Creation}
    	\State Faces with color $c$ continue to have the same color. Merged face is assigned one of the colors $c^\prime$ or $c^{\prime \prime}$. \Comment Faces that shared the deleted edge as common are merged.
    	\While{$e_{\ell}^{(i)}$ not reached} \Comment Twist Movement
    	\State Repeat steps $3$ and $4$ choosing the next edge along  the path $\pi_i$, see Figs.~\ref{fig:tm1},~\ref{fig:tm2}.
    	\EndWhile
    	\EndFor
    \end{algorithmic}
\end{algorithm}

\begin{remark}
If a third $c$-colored vertex from a given vertex $u$ is at is at a distance two, then multiple edges are introduced during the execution of Algorithm~\ref{alg:color-permuting-twists}.
If multiple edges are introduced, then remove the vertices incident on multiple edges and connect the resulting two-valent vertices.
\end{remark}

\begin{figure}[htb]
    \centering
    \begin{subfigure}{.225\textwidth}
       \centering
       \includegraphics[height = 3cm, width = 4cm]{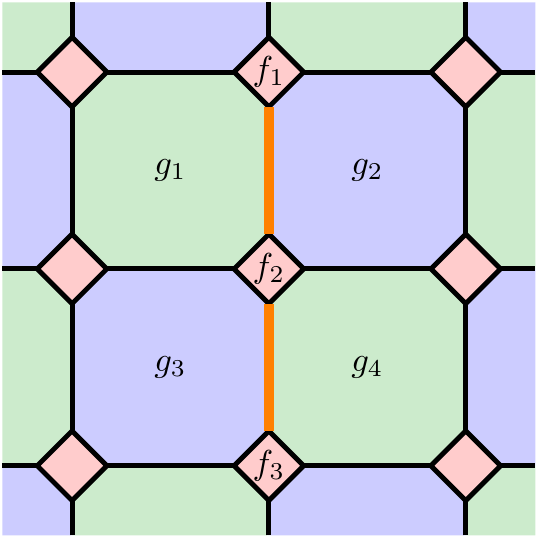}
       \subcaption{}
       \label{fig:tc0}
    \end{subfigure}
    ~
    \begin{subfigure}{.225\textwidth}
       \centering
       \includegraphics[height = 3cm, width= 4cm]{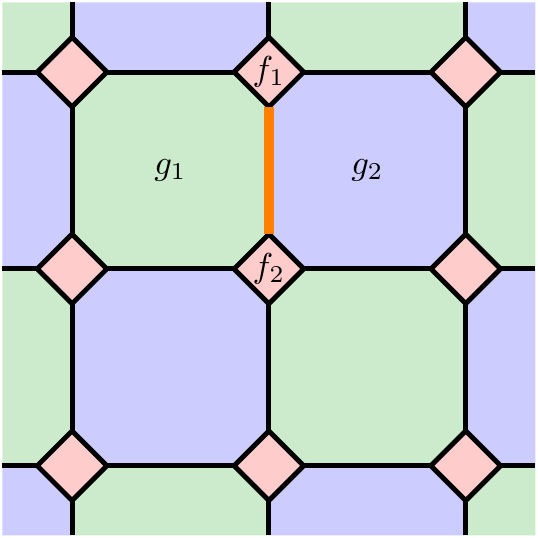}
       \subcaption{}
       \label{fig:tc1}
    \end{subfigure}
    ~
    \begin{subfigure}{.225\textwidth}
       \centering
       \includegraphics[height = 3cm, width = 4cm]{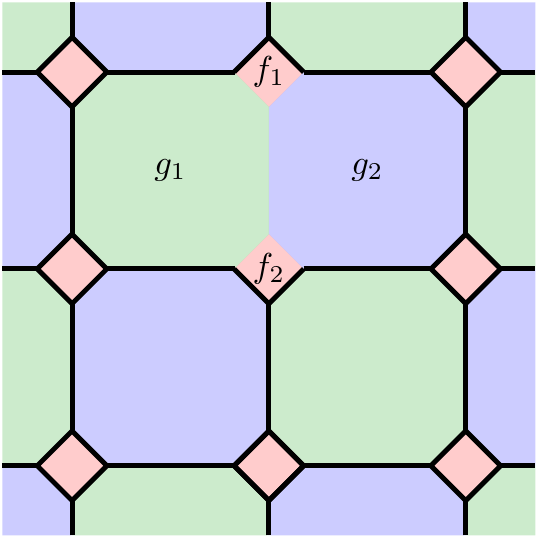}
       \subcaption{}
       \label{fig:tc2}
    \end{subfigure}
    ~
    \begin{subfigure}{.225\textwidth}
       \centering
       \includegraphics[height = 3cm, width = 4cm]{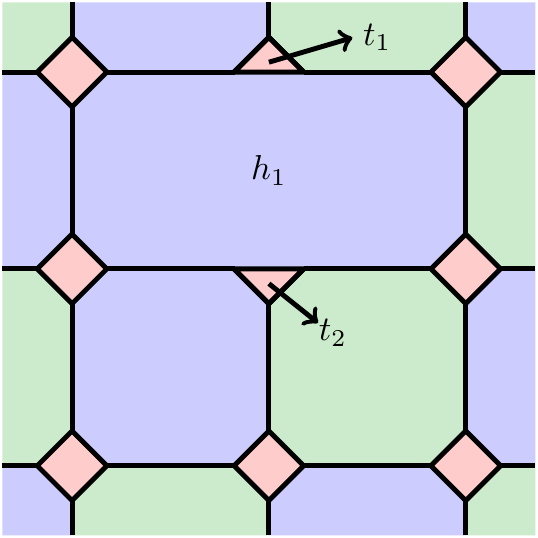}
       \subcaption{}
       \label{fig:tc3}
    \end{subfigure}
    ~
    \begin{subfigure}{.225\textwidth}
       \centering
       \includegraphics[height = 3cm, width = 4cm]{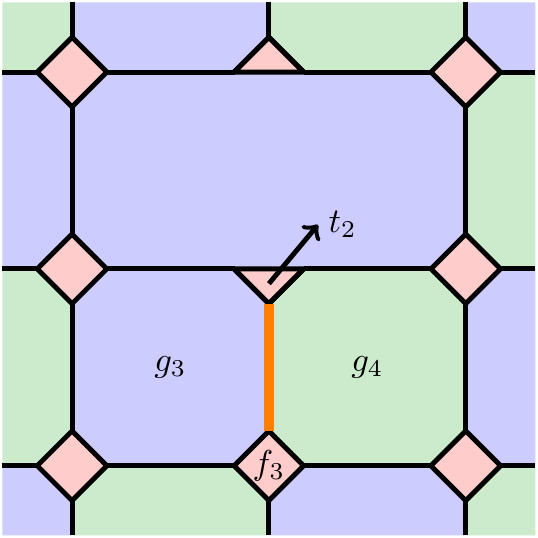}
       \subcaption{}
       \label{fig:tm1}
    \end{subfigure}
    ~
    \begin{subfigure}{.225\textwidth}
       \centering
       \includegraphics[height = 3cm, width = 4cm]{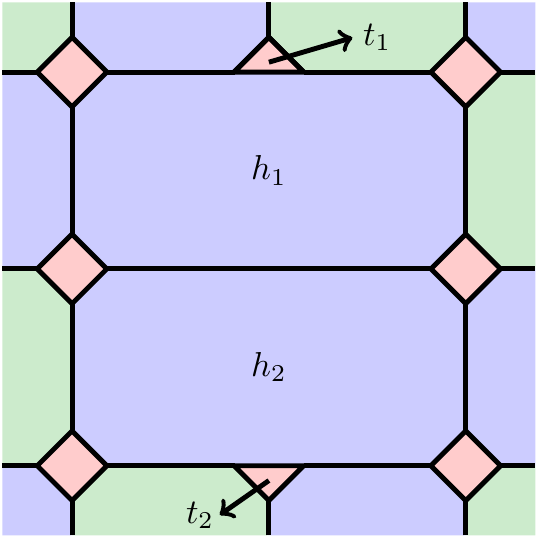}
       \subcaption{}
       \label{fig:tm2}
    \end{subfigure}
    \caption{Illustration of twist creation and movement in the square octagon lattice. (a) The path chosen for twist creation is shown in orange solid line and is of length two. (b) To create twists, remove the vertices that lie along the first edge of the path. (c) Two-valent vertices created in the lattice as a result of removing vertices. (d) Connect the nearest two valent vertices. Color of the new faces is as shown.  (e) To move twist, choose one outgoing edge of the same color as the twist. (f) Repeating the process of twist creation, the twist is moved to the next red face on the path.}
    \label{fig:tm}
\end{figure}

\begin{remark}
The process of twist creation can also be understood as deleting the orange edge in Fig.~\ref{fig:tc2} and then contracting one edge in each divalent vertex.
\end{remark}

\subsection{Stabilizers}
We need not modify a majority of stabilizers for color codes with color permuting twists.
Every face $f$ which is not a twist has an even number of edges.
Hence, two stabilizers are defined on such faces:
\begin{equation}
B_f^X = \prod_{v \in V(f)} X_v, \text{ and } B_f^Z = \prod_{v \in V(f)} Z_v 
\label{eqn:color-perm-nontwist-stab}
\end{equation}
However, on a twist face $\tau$, we cannot define two stabilizers as they will anti-commute. 

A $Y$-type stabilizer is defined on the twist face~\cite{Kesselring2018}:
\begin{equation}
		B_{\tau} = \prod_{v \in V(\tau)} Y_v.
		\label{eqn:color-perm-twist-stab}
\end{equation}

\noindent \emph{Consistency of stabilizer assignment.} If two faces are adjacent, then they share two common vertices.
To check commutation between adjacent face stabilizers, we have to check commutation on these vertices.
The stabilizers are all either $X$ type or $Z$ type.
If the Pauli operator on the common vertices is same, commutation follows.
On the other hand, if the Pauli operators on the common vertices are different (for instance, $X$ or $Z$ on one of the faces and $Y$ on the other face), then there are two anticommutations and the stabilizers commute.
Therefore, all stabilizers commute.

\noindent \emph{Stabilizer dependencies.}
Giving stabilizer dependencies for the general case is difficult.
We restrict to lattices where all twist faces are red.
Then the stabilizer generators have to satisfy the following constraints:
\begin{subequations}
\begin{eqnarray}
  \prod_{f \in \mathsf{F}_b} B_{f}^Z \prod_{f \in \mathsf{F}_g} B_{f}^Z    &=& I, \\
  \prod_{f \in \mathsf{F}_b} B_{f}^X \prod_{f \in \mathsf{F}_g} B_{f}^X  &=& I,
\end{eqnarray}
 \label{eqn:color-perm-stab-constraints}
\end{subequations}
where the product over blue faces includes blue modified faces and the outer unbounded blue face.
These constraints indicate the presence of two dependent stabilizers.
We take the dependent stabilizers to be the ones defined on the unbounded blue face.

\begin{proposition}
Algorithm~\ref{alg:color-permuting-twists} creates color permuting twists.
\end{proposition}

\begin{proof}
To prove that Algorithm~\ref{alg:color-permuting-twists} creates color permuting twists in conjunction with the
stabilizer assignment in Eqs.~\eqref{eqn:color-perm-nontwist-stab}~and~\eqref{eqn:color-perm-twist-stab}, it suffices to show that the algorithm introduces an edge between faces of different color while preserving trivalency of vertices.
When an anyon is moved along this edge, its color is permuted.
Let $u $ and $v$ be the vertices marked for deletion in the algorithm.
Let $N(u) = \{x_1, x_2, v\}$ and $N(v) = \{ y_1, y_2, u\}$.
Note that $e_1 = (x_1, x_2)$ and $e_2 = (y_1, y_2)$ do not exist before twist creation.
Due to the property of $2$-colex that the set of faces sharing an edge with any given face can be colored with two colors, the edges $e_1$ and $e_2$ connect faces of different color.
Algorithm~\ref{alg:color-permuting-twists} removes vertices $u$ and $v$ and introduces edges $e_1$ and $e_2$ thereby creating an edge between faces of different color.
Stabilizers on the unmodified faces are not changed.
Also, the stabilizer on the twist face is of $Y$ type.
Therefore, any operator $P_{x_1} P_{x_2}$, where $P \in \{X,Y,Z \} $ will commute with the twist stabilizer while moving the syndrome from face of one color to another.
Apart from these, the modified face in between twist has edges of different colors incident on it.
As a result, whenever an anyon enters and exits this faces through edges of different color, it color label is permuted.
Stabilizers defined on modified faces are of $Z$ and $X$ type.
Therefore, the error operators that move the syndrome across modified faces commute with the stabilizers defined on them.
\end{proof}

\begin{lemma}[Parity of color permuting twists]
Color permuting twists are created in pairs.
\label{lm:parity-color}
\end{lemma}

\begin{proof}
Let $\Gamma$ be the lattice with color permuting twists and let $\Gamma^\ast$ be its dual.
Note that color permuting twists are odd cycles and hence vertices corresponding to them have odd degree in the dual lattice.
Since the number of vertices with odd degree should be even in a graph \cite{Diestel} (and hence its embedding on a surface), color permuting twists are created in pairs.
\end{proof}

\noindent{\em $T$-lines and domain walls.}
When color permuting twists are created in the lattice, local three colorability around twist faces is destroyed.
As a consequence, we can find a sequence of edges, having the same color as the twist face, along which faces of the same color meet. On the (red) shrunk lattice these edges form a 
path terminating on the vertices corresponding to the twist faces.
We refer to this sequence of edges as a $T$-line following \cite{Bombin2011}.
For instance, if the twist face is red, then this path is formed by red edges and along this path blue and/or green adjacent faces meet.
An example of trivalent lattice with twists along with $T$-line is shown in Fig.~\ref{fig:hex-lattice}.
Existence of $T$-lines in the presence of color permuting twists is proved in Appendix~\ref{sec:t-line}.
Note that for this specific coloring of the lattice domain walls run parallel to the $T$-line.
A different coloring might be possible where the $T$-line crosses the domain wall as for instance happens
in case of surface codes with twists \cite{GowdaSarvepalli2020}. 
 
\begin{figure}[htb]
\centering
\includegraphics[height = 5.5cm,width = 6.5cm ]{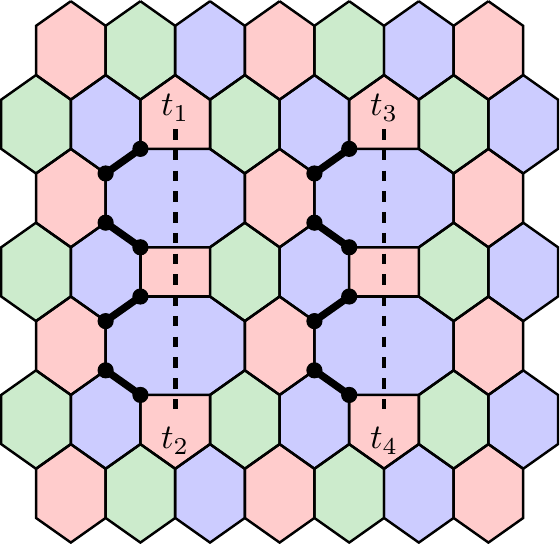}
\caption{Two pairs of color permuting twists in hexagonal lattice. Due to the presence of twists, there is a path, starting from one of the twists and ending on another, along which faces of same color (in this case, blue faces) meet. This path is shown as a solid line. Also, the domain wall is shown as dashed line.}
\label{fig:hex-lattice}
\end{figure}

\noindent \emph{Coding theoretic view of color permuting twist creation.} The process of color permuting twist creation can be understood in terms of classical code puncturing \cite{Huffman2003}.
In classical codes, code puncturing is done by deleting one or more (classical) bits and readjusting the parity check and generator matrices if necessary.
While creating twists in color codes, two qubits are removed from the system, see Fig.~\ref{fig:tc2}.
Suppose that stabilizers are written in symplectic form.
Let $f_1$ and $f_2$ be the top and bottom red faces in the figure and let $f_3$ and $f_4$ be the green and blue faces respectively.
Due to qubit removal, one qubit is removed from the stabilizers on $f_1$.
This leads to anti commutation among the stabilizers on $f_1$.
The anti commutation is rectified by dropping a stabilizer (either $Z$ or $X$ type). 
The same holds true for stabilizers on $f_2$.
Now, the stabilizers on $f_3$ and $f_4$ anti commute with stabilizer on $f_1$.
The resolution of this anticommutation comes from merging the corresponding $Z$ and $X$ stabilizers on $f_3$ and $f_4$.
In this process, we have deleted two columns, one row each of stabilizers of $f_1$ and $f_2$ and merged the rows of $f_3$ and $f_4$.
This process is exactly similar to code puncturing in classical codes.
Therefore, color permuting twist creation is the analogue of code puncturing in quantum codes.

For a general 2-colex into which color permuting twists are introduced the following result holds. 
\begin{lemma}[Encoded qubits for color permuting twists]
A 2-colex with $t$ color permuting twists encodes $(t - 2)$ qubits where $t$ is even.
\label{lm:encoded-qubits-cp}
\end{lemma}
\begin{proof}
Note that in lattice containing color permuting twists, all vertices are trivalent.
Hence, we have $2e = 3v$, where $e$ and $v = n$ are the number of edges and vertices (qubits) of the lattice respectively.
Also, $v + f - e = 2$, where $f$ is the number of faces in the lattice.
Therefore, we get, $f = \frac{v}{2} + 2 = \frac{n}{2} + 2$.
We count two stabilizers for every face, including twists.
Subtracting the dependent stabilizers, which are $t$ in number 
 gives the number of remaining stabilizers to be $2f - t = n - t + 4$.
Ignoring the outer unbounded face stabilizers (see Equations~\eqref{eqn:color-perm-stab-constraints}), 
we obtain the independent stabilizer generators to be $s = n - 2\left(\frac{t}{2} - 1\right)$.
Hence, the number of encoded qubits is $2(\frac{t}{2} - 1)$.
\end{proof}

\begin{remark}
The number of encoded qubits here is twice that of surface codes with twists~\cite{GowdaSarvepalli2020}.
The reason is, with the introduction of a twist pair in surface codes, one stabilizer was removed and hence one new encoded qubit.
However, in color codes with every twist pair, two stabilizers are removed and hence we get twice the number of encoded qubits.
\end{remark}

\subsection{Pauli operators as strings on a 2-colex with color permuting twists} 
In this subsection, we give a string representation of Pauli operators on a 2-colex in the presence of color permuting twists.
Consider Fig.~\ref{fig:Qb-cc-pauli-string-10} where Pauli operator $Z$ is applied on the vertices of the edge.
The red face on the top is twist, green and blue faces on the right and left of twist are unmodified faces and the bottom face with blue color is the modified face.
The Pauli operators create two $ry$ syndromes on the twist face, $bx$ and $gx$ on the normal blue and green faces respectively and $gx$, $bx$ on the modified face
(the part of the modified face to the right of the domain wall supports anyons with green color label).
Anyons on the twist face annihilate each other and hence the string is continuous in the twist face.
The Pauli operator commutes with the stabilizers on the modified face and therefore string is continuous in modified face too, see Fig.~\ref{fig:Qb-cc-pauli-string-11}.
String  representation of the Pauli operator in Fig.~\ref{fig:Qb-cc-pauli-string-10} is given in Fig.~\ref{fig:Qb-cc-pauli-string-12}.
Note that the string changes color as it crosses the domain wall.

\begin{figure}[htb]
    \centering
    \begin{subfigure}{.225\textwidth}
        \centering
        \includegraphics[scale = 1]{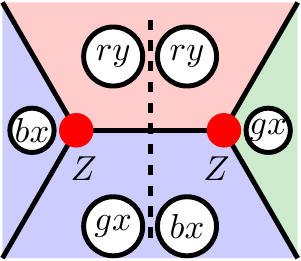}
        \subcaption{}
        \label{fig:Qb-cc-pauli-string-10}
    \end{subfigure}
    ~
    \begin{subfigure}{.225\textwidth}
        \centering
        \includegraphics[scale = 1]{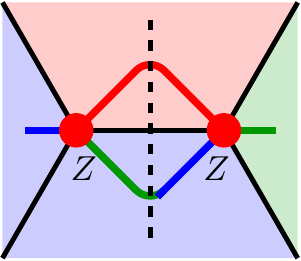}
        \subcaption{}
        \label{fig:Qb-cc-pauli-string-11}
    \end{subfigure}
    ~
    \begin{subfigure}{.225\textwidth}
        \centering
        \includegraphics[scale = 1]{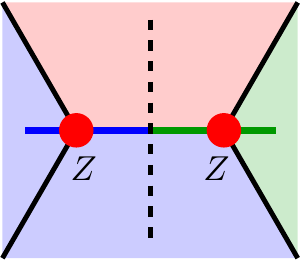}
        \subcaption{}
        \label{fig:Qb-cc-pauli-string-12}
    \end{subfigure}
    ~
    \begin{subfigure}{.225\textwidth}
        \centering
        \includegraphics[scale = 1]{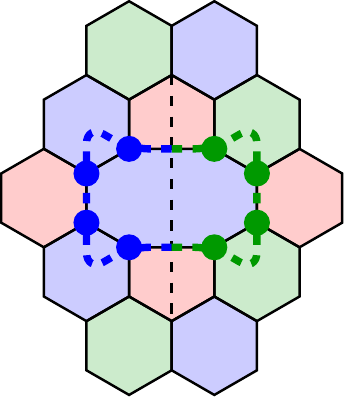}
        \subcaption{}
        \label{fig:Qb-color-perm-modified-string-X}
    \end{subfigure}
    ~
    \begin{subfigure}{.225\textwidth}
        \centering
        \includegraphics[scale = 1]{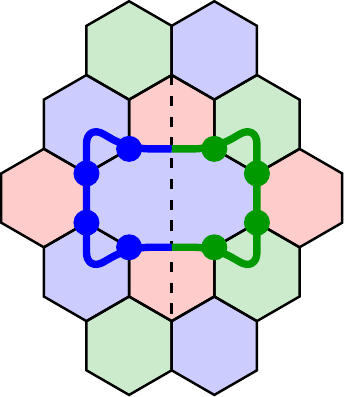}
        \subcaption{}
        \label{fig:Qb-color-perm-modified-string-Z}
    \end{subfigure}
    \caption{Representing Pauli operators as strings in the presence of color permuting twists. (a) Syndromes created as a result of applying $Z$. Upper red face is twist and the lower face with mixed color is modified face. (b) Syndromes on twist vanish. On the modified face if we move the syndrome $gx$ across the domain wall it becomes $bx$ and annihilates the other $bx$ syndrome giving a zero syndrome. (c) The string corresponding to the operator in Fig.~\ref{fig:Qb-cc-pauli-string-10} is of mixed color. The color change happens at the domain wall. (d) $X$ type stabilizer on modified face. (e) $Z$ type stabilizer on modified face.}
    \label{fig:Qb-color-perm-modified-string}
\end{figure}

Stabilizers on modified faces cannot be represented as strings with one color, see Fig~\ref{fig:Qb-color-perm-modified-string-X} and Fig.~\ref{fig:Qb-color-perm-modified-string-Z}.
For the same string, both $X$ and $Z$ stabilizers are defined.
Note the strings change color as they crosses the domain wall.

As in the case with charge permuting twists, logical operators are closed strings encircling an even number ($ \ge 2$) of twists.
(A string encircling an odd number of twists will be an open string.)
The above mentioned logical operators are not generated by stabilizers.
The reason being one can find another string encircling a pair of twists that anticommutes with the given string, see Fig.~\ref{fig:hexagon-lattice-LO-X} and Fig.~\ref{fig:hexagon-lattice-LO-Z}.
The logical operators for one of the encoded qubits is shown in Fig.~\ref{fig:hexagon-lattice-LO}.
Logical $X$ operator is the string which encircles twists not created together, see Fig.~\ref{fig:hexagon-lattice-LO-X}.
Logical $Z$ operator is the string that encircles twists created together, see Fig.~\ref{fig:hexagon-lattice-LO-Z}.
Note that with every pair of twists introduced, the number of encoded qubits increases by two.

\begin{figure}[htb]
    \centering
    \begin{subfigure}{.45\textwidth}
        \centering
        \includegraphics[height = 5.5cm,width = 6.5cm]{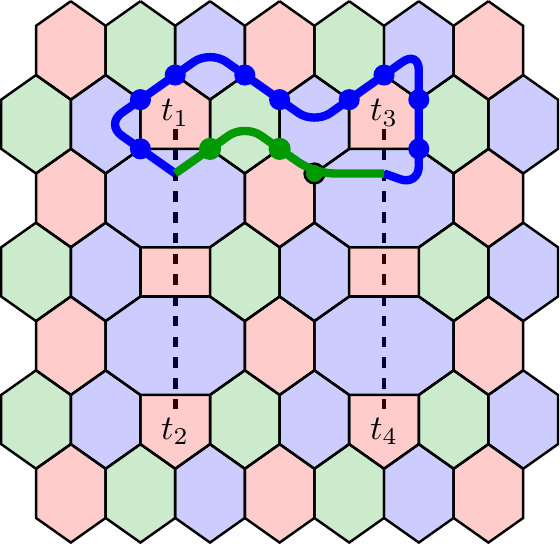}
        \subcaption{}
        \label{fig:hexagon-lattice-LO-X}
    \end{subfigure}
    ~
    \begin{subfigure}{.45\textwidth}
        \centering
        \includegraphics[height = 5.5cm,width = 6.5cm]{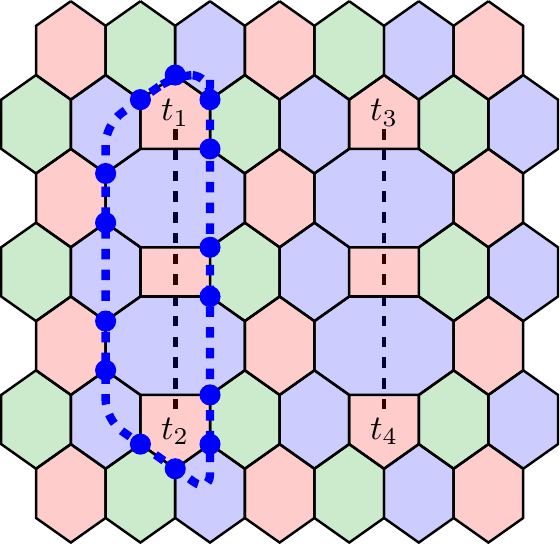}
        \subcaption{}
        \label{fig:hexagon-lattice-LO-Z}
    \end{subfigure}
    \caption{Logical operators for an encoded qubit depicted on lattice. (a) Logical $X$ operator depicted on lattice. Blue circles and green circles with black boundary represent Pauli operator $Z$. (b) Logical $Z$ operator depicted on lattice. Blue circles represent Pauli operator $X$.}
    \label{fig:hexagon-lattice-LO}
\end{figure}

Henceforth, we consider all twists to be of the same color.
The reason for this is that when twist pairs are of different color, logical $X$ operator has the form shown in Fig.~\ref{fig:diff-color-twists}.
To avoid such self-intersecting logical operators, we choose all twists to be of the same color.
Logical operators for encoded qubits in color codes with color permuting twists is given in Fig.~\ref{fig:lo_color}.
The string encircling twists $t_1$ and $t_2$ is the logical $Z$ operator and the one enclosing twists $t_1$ and $t_3$ is the logical $X$ operator.
When two strings of same color intersect they overlap in even number of locations, therefore
they commute. 
Only when strings of different color intersect they overlap in one location, therefore they anticommute. 
From this we conclude that the operators in Fig.~\ref{fig:lo_color} anticommute. 
Even when the support of both logical operators intersect where the colors are same, the size of the common support is always an even number.
We choose to use $\lfloor t/3 \rfloor$ encoding for the same reason as that for charge permuting twists.
The canonical form of logical operators for $\lfloor t/3 \rfloor$ encoding is shown in Fig.~\ref{fig:lo_color_canonical}.

\begin{figure}[htb]
\centering
\begin{subfigure}{.45\textwidth}
    \centering
    \includegraphics[height = 3cm, width = 3cm]{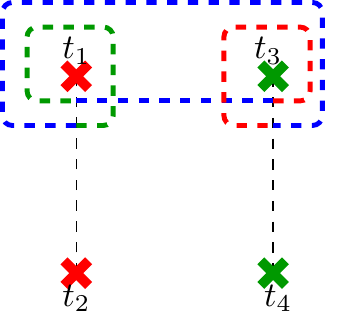}
    \subcaption{}
    \label{fig:diff-color-twists}
\end{subfigure}
~
\begin{subfigure}{.45\textwidth}
\centering
\includegraphics[height = 3cm, width = 6cm]{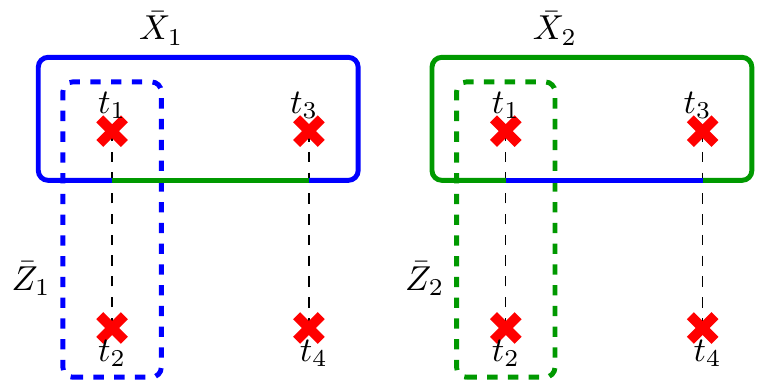}
\subcaption{}
\label{fig:lo_color}
\end{subfigure}
~
\begin{subfigure}{.45\textwidth}
\centering
\includegraphics[height = 3cm, width = 4.5cm]{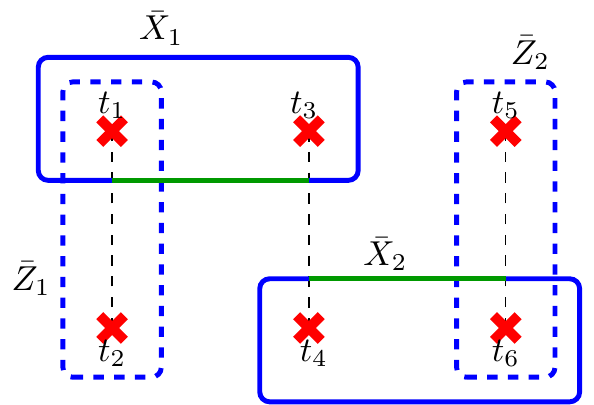}
\subcaption{}
\label{fig:lo_color_canonical}
\end{subfigure}
\caption{Logical operators for encoded qubits in lattices with color permuting twists. (a) Logical $X$ operators when twist pairs are of different color. (b) Logical operators for encoded qubits in lattice having color permuting twists. Operators $\Bar{Z}_1$ and $\Bar{Z}_2$ have $Z$ and $X$ on their support respectively whereas operators $\Bar{X}_1$ and $\Bar{X}_2 $ have $X$ and $Z$ on their support respectively. (c) Canonical logical operators for encoded qubits in lattice having color permuting twists for $\lfloor t/3 \rfloor $.}
\label{fig:color_LO}
\end{figure}

\begin{remark}
With $\lfloor t/3 \rfloor $ encoding, we get $\left(t - 2\right) - \left\lfloor \frac{t}{3} \right\rfloor$ gauge qubits.
\end{remark}

\section{Clifford gates using charge permuting twists}
\label{sec:gates-charge-permuting}
In this section, we show that encoded Clifford gates can be realized with charge permuting twists by braiding alone.
We use $\left \lfloor t / 3 \right \rfloor$ encoding while realizing gates, see Fig~\ref{fig:lo_charge_canonical}.
{Braiding twists $t_i$ and $t_j$ induces a permutation $i\leftrightarrow j$. 
This naturally leads to a transformation of closed strings enclosing the twists. 
Denote a string of color $c$ enclosing twists $t_i$ and $t_j$ by $\mathcal{W}_{i,j}^{c}$.}
{After braiding twists $t_i$ and $t_j$,  the string $\mathcal{W}_{i,j}^{c}$ encircling them is unchanged. 
However, a string $\mathcal{W}_{i,k}^{c}$ encircling twists $t_i$ and other twist $t_k$ after braiding (of $t_i$ and $t_j$) is mapped
to  $\mathcal{W}_{j,k}^c$. In other words, after braiding $\mathcal{W}_{i,k}^c \mapsto \mathcal{W}_{j,k}^c$.}
Let $t_1$, $t_2$ and $t_3$ be twists as shown in Fig.~\ref{fig:braid-single-qubit}.
The list of logical operators and corresponding strings is given in Table~\ref{tab:string_LO}.
In some case a logical operator has an equivalent representation obtained by adding  a gauge operator, see Fig.~\ref{fig:equivalent_LO}. 
However, strings corresponding to logical $\bar{Y}_1 (\bar{Y}_2) $ operator do not have equivalent strings of different color of the form $\mathcal{W}_{2,3}^c (\mathcal{W}_{4,5}^c)$.

\begin{table}[htb]
    \centering
    \begin{tabular}{l|c|c}
    \hline
    \hline
    String & Equivalent string & Logical Operator\\
    \hline
         $\mathcal{W}_{1,2}^{b}$ & $\mathcal{W}_{1,2}^{r}$ & $\bar{Z}_1$ \\
        $\mathcal{W}_{1,3}^{g}$ & $\mathcal{W}_{1,3}^{r}$ & $\bar{X}_1$ \\
        $\mathcal{W}_{2,3}^{r}$ & $-$ & $\bar{Y}_1$ \\
        $\mathcal{W}_{5,6}^{b}$ & $\mathcal{W}_{5,6}^{r}$& $\bar{Z}_2$ \\
        $\mathcal{W}_{4,6}^{g}$ & $\mathcal{W}_{4,6}^{r}$ & $\bar{X}_2$ \\
        $\mathcal{W}_{4,5}^{r}$ & $-$ & $\bar{Y}_2$ \\
        \hline
    \end{tabular}
    \caption{Strings and the corresponding logical operator associated with them. Note that the strings corresponding to logical $Y$ operators are dependent on those of logical $X$ and $Z$ operators. }
    \label{tab:string_LO}
\end{table}

\begin{remark}
The string $\mathcal{W}_{2,3}^r$ corresponding to $\bar{Y}_1$ is not equivalent to $\mathcal{W}_{2,3}^g$, $\mathcal{W}_{2,3}^b$
unlike $\bar{X}_1 $ and $\bar{Z}_1 $. 
It turns out that $\mathcal{W}_{2,3}^g$ and $\mathcal{W}_{2,3}^b$ correspond to operators $\bar{X}_1 \bar{Z}_2$ and $\bar{Z}_1 \bar{X}_2$ respectively.
\end{remark}

\begin{theorem}
[Single qubit Clifford gate]
Let $t_1$, $t_2$ and $t_3$ be twists that encode a logical qubit and let  $\mathcal{W}_{1,2}^{b}$,  $\mathcal{W}_{1,3}^{g}$ and  $\mathcal{W}_{2,3}^{r}$ be the strings corresponding to logical operators $\Bar{Z}$, $\Bar{X}$ and $\Bar{Y}$ respectively.
Then,
\begin{compactenum}[i)]
    \item Braiding $t_1$ and $t_2$ realizes phase gate ($Z$ rotation by $\pi /2$).
    \item Braiding $t_1$ and $t_3$ realizes $X$ rotation by $\pi /2$.
\end{compactenum}
\label{thm:single-qb-Clifford}
\end{theorem}
\begin{proof}
The encoding used is as shown in Fig.~\ref{fig:lo_charge_canonical}.
The logical $Y$ operator (up to gauge operators) is shown in Fig.~\ref{fig:Qb-cc-logY-4}.
After braiding twists $t_1$ and $t_2$ we have, $t_1 \leftrightarrow t_2$, $t_3 \rightarrow t_3$ and $t_4 \rightarrow t_4$. Using this we get the transformation given below. 
\begin{eqnarray*}
  \begin{array}{ccc}
      \mathcal{W}^{b}_{1,2} \mapsto \mathcal{W}_{1,2}^{b} & \vert & \bar{Z} \mapsto \bar{Z}\\
      \mathcal{W}^{g}_{1,3} \mapsto \mathcal{W}_{2,3}^{g} & \vert  & \bar{X} \mapsto \bar{Y}\\
      \mathcal{W}_{2,3}^{r} \mapsto \mathcal{W}_{1,3}^{r}  & \vert & \bar{Y} \mapsto \bar{X} \\
  \end{array}
\end{eqnarray*}

The string $\mathcal{W}_{1,2}^{b}$ is unchanged since the twists encircled by it are braided.
Transformation of $\mathcal{W}_{1,3}^{g}$ is shown in Appendix~\ref{sec:charge-braiding} and transformation of $\mathcal{W}_{2,3}^{r}$  to $\mathcal{W}_{1,3}^{r}$ is straightforward.
Similarly, braiding $t_1$ and $t_3$, we get the following transformation.
\begin{eqnarray*}
  \begin{array}{ccc}
       \mathcal{W}_{1,3}^{g} \mapsto \mathcal{W}_{1,3}^{g}    & \vert &  \bar{X} \mapsto \bar{X},\\
  \mathcal{W}_{1,2}^{b} \mapsto \mathcal{W}_{2,3}^{b}  & \vert & \bar{Z} \mapsto \bar{Y} \\
  \mathcal{W}_{2,3}^{r} \mapsto \mathcal{W}_{1,2}^{r} & \vert & \bar{Y} \mapsto \bar{Z} \\
  \end{array}
\end{eqnarray*}
Hence, the theorem  follows.
\end{proof}

Hadamard gate can be realized by $X$ rotation and phase gates.
Hence $Z$ and $X$ rotations by $\pi / 2 $ suffice to realize single qubit Clifford gates by braiding.

\begin{figure}[htb]
    \centering
    \begin{subfigure}{.225\textwidth}
        \centering
        \includegraphics[scale = 1]{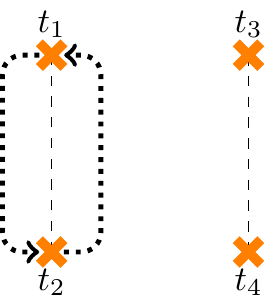}
        \subcaption{}
        \label{fig:braid-single-qubit-rz}
    \end{subfigure}
    ~
    \begin{subfigure}{.225\textwidth}
        \centering
        \includegraphics[scale = 1]{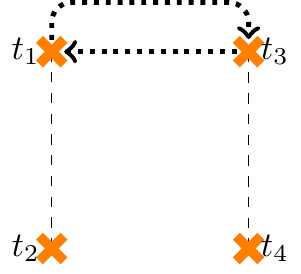}
        \subcaption{}
        \label{fig:braid-single-qubit-rx}
    \end{subfigure}
    \caption{Braiding twists to implement single qubit Clifford gates. (a) To implement a Pauli $Z$ rotation (i.e. Phase gate), we have to braid the twists that are encircled by the logical $Z$ operator viz $t_1$ and $t_2$. (b) Pauli $X$ rotation by $\pi / 2$ is implemented by braiding twists $t_1$ and $t_3$ which are encircled by the string corresponding to logical $X$ operator.}
    \label{fig:braid-single-qubit}
\end{figure}

To complete the Clifford group, we need an entangling gate.
We prove below that braiding the twist pair shared by two logical qubits accomplishes this.
Let $t_{1}$, $t_{2}$ and $t_{3}$ be twists that encode the logical qubit $a$ and let $t_{4}$, $t_{5}$ and $t_{6}$ be twists that encode the logical qubit $b$ in the same block as shown in Fig.~\ref{fig:braid-entangling-adjacent}.
The operator corresponding to string $\mathcal{W}_{3,4}^{b}$ encircling twists $t_{3}$ and $t_{4}$ is the product of operators encircling twists $t_{1}$, $t_{2}$ and $t_{5}$, $t_{6}$  i.e. $\mathcal{W}_{3,4}^{b} = \mathcal{W}_{1,2}^{b} \mathcal{W}_{5,6}^{b}$.
This can be seen easily as the logical operators in the form of closed strings can be opened up and running between boundaries.
This is possible as the logical operators encircling twists created in pairs have support on blue strings which can be deformed to end on the unbounded blue face.
Then, by adding stabilizers, the open strings can be made to encircle twists $t_{3}$ and $t_{4}$.
Similarly, for qubits not in the same block, see Fig.~\ref{fig:braid-entangling-nonadjacent}, the string encircling twists $t_3$, $t_4$ and $t_9$, $t_{10}$ can be deformed in similar way to that in surface codes~\cite{GowdaSarvepalli2020}.
The twists encoding qubits is given in Table~\ref{tab:twists_encoding}.
\begin{table}[htb]
    \centering
    \begin{tabular}{c|c}
     \hline
    \hline
       Encoded qubit  & Twists used \\
       \hline
        $a$ & $t_1$, $t_2$, $t_3$ \\
        \hline
        $b$ & $t_7$, $t_8$, $t_9$\\
        \hline
        $c$ & $t_4$, $t_5$, $t_6$\\
        \hline
        $d$ & $t_{10}$, $t_{11}$, $t_{12}$\\
        \hline
    \end{tabular}
    \caption{Encoded qubits and the twists used for encoding.}
    \label{tab:twists_encoding}
\end{table}

\begin{figure}[htb]
    \centering
    \begin{subfigure}{.225\textwidth}
        \centering
        \includegraphics[scale = .85]{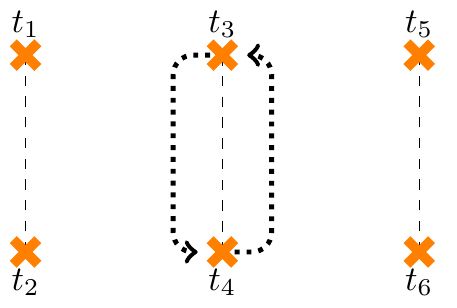}
        \subcaption{}
        \label{fig:braid-entangling-adjacent}
    \end{subfigure}
    ~
    \begin{subfigure}{.225\textwidth}
        \centering
        \includegraphics[scale = .85]{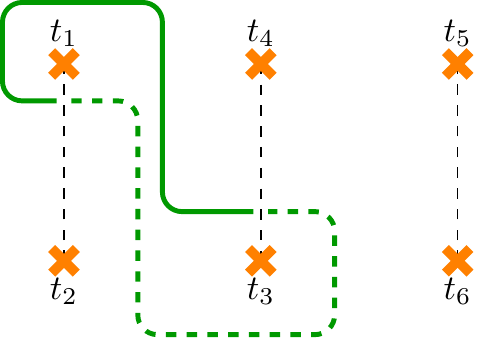}
        \subcaption{}
        \label{fig:braid-entangling-adjacent-log-X}
    \end{subfigure}
    ~
    \begin{subfigure}{.45\textwidth}
        \centering
        \includegraphics[scale = .9]{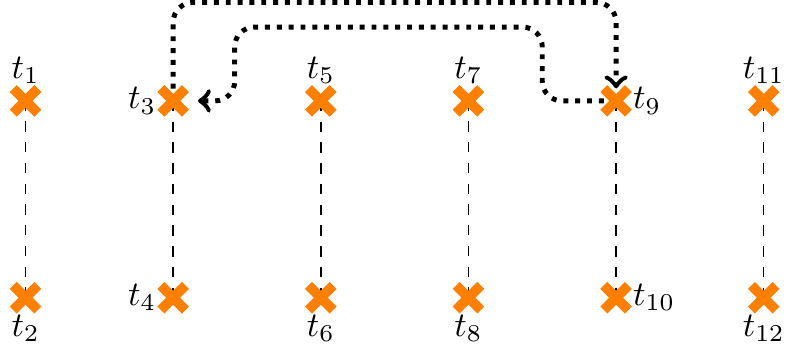}
        \subcaption{}
        \label{fig:braid-entangling-nonadjacent}
    \end{subfigure}
    \caption{Realizing entangling gate between qubits in the same block and different block. (a) To realize entangling gate between two encoded qubits in the same block, braid the twist pair as shown counterclockwise. (b) After braiding $t_3$ and $t_4$, $\Bar{X}_1$ is modified as shown. This string can be expressed as a combination of string encircling twists $t_3$ and $t_4$ and $\Bar{X}_1$ before braiding. (c) An entangling gate between two encoded qubits in different blocks is realized by braiding twists as shown.}
    \label{fig:braid-entangling}
\end{figure}

\begin{theorem}[Entangling gate]
Controlled-Z gate is realized up to phase gate on control qubit $a$ and target qubit $b$ by the following braiding.
\begin{compactenum}[i)]
    \item Braid $t_3$ and $t_4$ for qubits in the same block.
    \item Braid $t_3$ and $t_9$ for qubits in different block.
\end{compactenum}
\label{thm:multi-qb-entangling}
\end{theorem}

\begin{proof}
We prove for the qubits in the same block.
Extension to qubits from nonadjacent block is straightforward.
Encoding used is as shown in Fig.~\ref{fig:lo_charge_canonical}.
After braiding $t_{3}$ and $t_{4}$, we have $t_{3} \leftrightarrow t_{4}$, $t_{m} \rightarrow t_{m}$, $m = 1,2,5,6$. 
We have the following transformation.
\begin{eqnarray*}
  \begin{array}{ccc}
   \mathcal{W}_{3,4}^{b} \mapsto \mathcal{W}_{3,4}^{b}    & \vert &  \bar{Z}_1\bar{Z}_2 \mapsto \bar{Z}_1\bar{Z}_2,\\
  \mathcal{W}_{1,2}^{b} \mapsto \mathcal{W}_{1,2}^{b}  & \vert &  \bar{Z}_1 \mapsto \bar{Z}_1 \\
  \mathcal{W}_{1,3}^{g} \mapsto \mathcal{W}_{1,4}^{g}  & \vert & \bar{X}_1 \mapsto \bar{X}_1 \bar{Z}_1 \bar{Z}_2 = \bar{Y}_1 \bar{Z}_2\\
  \mathcal{W}_{5,6}^{b} \mapsto \mathcal{W}_{5,6}^{b}  & \vert &  \bar{Z}_2 \mapsto \bar{Z}_2 \\
  \mathcal{W}_{4,6}^{g} \mapsto \mathcal{W}_{3,6}^{g}  & \vert & \bar{X}_2 \mapsto \bar{X}_2 \bar{Z}_1 \bar{Z}_2 = \bar{Z}_1 \bar{Y}_2\\
  \end{array}
\end{eqnarray*}

In the third and the last transformations, the deformed strings can be expressed as $\mathcal{W}_{1,3}^r \mathcal{W}_{3,4}^b$ and $\mathcal{W}_{3,4}^b \mathcal{W}_{4,6}^r$ respectively, see Fig.~\ref{fig:Qb-cc-logY} in Appendix~\ref{sec:charge-braiding}.
This transformation is similar to that of combining $\Bar{Z}$ and $\Bar{X}$ to give $\Bar{Y}$.
By noting the logical operator transformation on the right side, we can conclude that braiding twists $t_3$ and $t_4$ results in controlled-$Z$ gate up to a phase gate on control and target qubits.
\end{proof}

After performing the braiding as indicated in Theorem~\ref{thm:multi-qb-entangling}, performing phase gate on control and target qubits will realize controlled-$Z$ gate between them.
Phase gate is performed by braiding the twist pair encircled by the logical $Z$ operator.
The procedure to realize the entangling gate stated in Theorem~\ref{thm:multi-qb-entangling} between the other encoded qubit pairs is given in Table~\ref{tab:entangling_braid_nonadjacent}.
\begin{table}[htb]
    \centering
    \begin{tabular}{c|c}
    \hline
    \hline
         Encoded qubits & Braid to be performed \\
         \hline
         $(a,d)$ & $(t_{3},t_{10})$ \\
         \hline
         $(b,c)$ &  $(t_{4},t_{9})$\\
         \hline
          $(b,d)$ & $(t_{4},t_{10})$\\
          \hline
    \end{tabular}
    \caption{The entangling gate between encoded qubits is realized by performing the braiding indicated. The first qubit in the pair $(a,d)$, viz. $a$, is the control qubit and the other is the target qubit. Controlled-$Z$ is realized by performing phase gate on both control and target qubits.}
    \label{tab:entangling_braid_nonadjacent}
\end{table}

\begin{theorem}
In a color code with charge permuting twists, Clifford gates can be realized by braiding alone.
\end{theorem}

\begin{proof}
Follows from Theorem~\ref{thm:single-qb-Clifford} and Theorem~\ref{thm:multi-qb-entangling}.
\end{proof}

We now summarize the braiding protocols for  realizing Clifford gates in  Table~\ref{tab:gate_braiding}.
\begin{table}[htb]
    \centering
    \begin{tabular}{c|c}
    \hline
    \hline
         Gate & Twists to be braided \\
         \hline
         $R_Z(\pi / 2)$ & $t_1$ and $t_2$ (Theorem~\ref{thm:single-qb-Clifford}) \\
         \hline
         $R_X(\pi / 2)$ & $t_1$ and $t_3$ (Theorem~\ref{thm:single-qb-Clifford}) \\
         \hline
         Adjacent CZ &  $t_3$ and $t_4$, $t_1$ and $t_2$, $t_5$ and $t_6$ (Theorem~\ref{thm:multi-qb-entangling})\\
         \hline
         Nonadjacent CZ & $t_3$ and $t_9$, $t_1$ and $t_2$, $t_7$ and $t_8$ (Theorem~\ref{thm:multi-qb-entangling})\\
         \hline
    \end{tabular}
    \caption{Braiding protocol for Clifford gates.}
    \label{tab:gate_braiding}
\end{table}

\section{Clifford Gates using color permuting twists}
\label{sec:gates-color-permuting}

In this section, we describe the implementation of Clifford gates using color permuting twists.
We use four color permuting twists to encode a logical qubit, see Fig.~\ref{fig:Qb-cc-four-twist-encoding-color}.

For realizing phase and Hadamard gates, we use Pauli frame update~\cite{Hastings2015}.
Pauli frame update is done classically.
In this procedure, one has to keep track of the Pauli frame for each encoded qubit i.e. one has to maintain the information whether the canonical logical operator labels are exchanged.
Phase gate is realized by interchanging the labels of $X$ and $Y$ logical operators and for realizing Hadamard gate, the labels of $X$ and $Z$ logical operators are interchanged.
If a $Z$ measurement is to be done after Hadamard gate, then after Pauli frame update, we take into account the Pauli frame update and measure the canonical $X$ logical operator.

We implement CNOT gate by making use of an additional ancilla qubit and holes~\cite{Terhal2015, Brown2017}.
This is done by joint $X$ parity measurement and joint $Z$ parity measurement with an ancilla qubit.
The protocol is given in Table~\ref{tab:cc_cnot}.

\begin{table}[htb]
    \centering
    \begin{tabular}{ll}
    \hline
    & Protocol for implementing CNOT gate~\cite{Terhal2015} \\
    \hline 
        (1) & Prepare ancilla in the state $|0\rangle$. \\
         (2) & Perform joint $X$ parity measurement on ancilla and target qubit. \\
         (3) & Perform joint $Z$ parity measurement on ancilla and control qubit.\\
         (4) & Perform Hadamard on ancilla qubit and measure it in $X$ basis.\\
         (5) & Measure ancillas and apply correction as given in Equation~\ref{eqn:CNOT}.\\
         \hline
    \end{tabular}
    \caption{Protocol for implementing CNOT gate by joint measurements with an ancilla.}
    \label{tab:cc_cnot}
\end{table}

Suppose that $|c\rangle$ and $|t\rangle$ are the states of control and target qubits with $c = 0,1$ and $t = 0,1$.
Let the measurement outcome of joint $X$ parity measurement be $m_{xx}$ and those of joint $Z$ parity measurement and ancilla measurement be $m_{zz}$ and $m_{x}$ respectively.
Then it can be shown that at the end of the protocol the initial state is modified as~\cite{Terhal2015}
\begin{equation}
    |c\rangle |t\rangle |0\rangle \rightarrow |c\rangle Z^{m_{xx}} X^{m_{zz}}|c + t\rangle Z^{m_{x}} |+\rangle.
    \label{eqn:CNOT}
\end{equation}
Ancilla qubit is disentangled by measuring in the $X$ basis.

\begin{figure}[htb]
    \centering
    \begin{subfigure}{.15\textwidth}
        \centering
        \includegraphics[scale = .75]{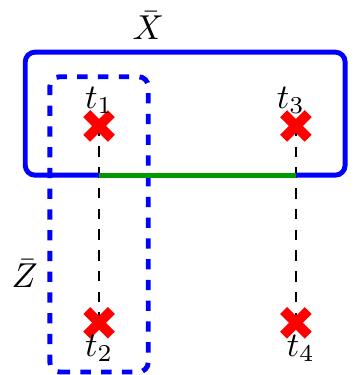}
        \subcaption{}
      \label{fig:Qb-cc-four-twist-encoding-color}
    \end{subfigure}
    ~
    \begin{subfigure}{.15\textwidth}
        \centering
        \includegraphics[scale = .75]{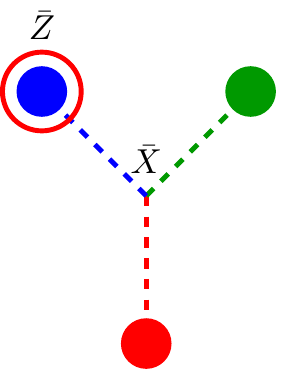}
        \subcaption{}
      \label{fig:primal-qubit-hole}
    \end{subfigure}
    ~
    \begin{subfigure}{.15\textwidth}
        \centering
        \includegraphics[scale = .75]{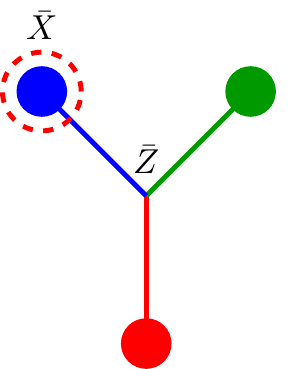}
        \subcaption{}
      \label{fig:hole-encoding}
    \end{subfigure}
\caption{Four twist encoding that used red color permuting twists. Also shown is the dual logical qubit encoded with a triple of holes. (a) Four twist encoding used for realizing encoded gates with color permuting twists. (b) Three holes are used to encode a logical (primal) qubit. (c) Three holes are used to encode a logical (dual) qubit.}
\label{fig:twist-hole-encoding}
\end{figure}

\noindent \emph{Motivation for using holes.} The encoded CNOT gate protocol given in Table~\ref{tab:cc_cnot} involves joint parity measurement of logical $Z$ and logical $X$ operators.
These operators are high weight owing to the separation between the twists.
(We keep the twist separation large to make the distance large.)
Therefore, measuring the the joint operators by doing gates between physical qubits in the support of logical operators and ancilla qubit fault-tolerant will be quite complex.
However, an alternate scheme exists for doing fault-tolerant measurement of these operators.
This alternate scheme in the context of surface codes involves creating holes, braiding them around twists and measuring them~\cite{Brown2017}. 
We adapt this technique to the  color permuting twists.

\noindent \emph{Holes in color code lattices.}
Here we briefly review holes in color code lattices.
For a more detailed treatment of holes in color codes, we refer the reader to Ref.~\cite{Fowler2011}.
A hole is a region in the lattice where stabilizers are not measured.
Suppose that we have to introduce a red hole in the lattice.
We choose a red face and do not measure the stabilizers defined on it.
To expand the hole, we take the vertices on one of the edges incident on the face.
Then $ZZ$ and $XX$ measurements are performed on the qubits on the vertices of the chosen edge.
Using this procedure, a hole can be made arbitrarily big.
We use triple hole encoding~\cite{Fowler2011}.
Two more twists, one each of green and blue colors, are created by the same procedure.
A logical qubit encoded in a hole is shown in Fig.~\ref{fig:primal-qubit-hole} and Fig.~\ref{fig:hole-encoding}.
Logical operators are the strings encircling hole and the string net (in the shape of T).
Depending on whether we choose $Z$ or $X$ logical operator to encircle the hole, we get primal and dual qubits respectively.
The encoding shown in Fig.~\ref{fig:primal-qubit-hole} is that of a primal qubit whereas the encoding in Fig.~\ref{fig:hole-encoding} is that of a dual qubit.

\subsection{Hole-twist braiding}
We now show how logical operators are modified by braiding a hole around twists. 
Three hole encoding~\cite{Fowler2011} is used, see Fig.~\ref{fig:primal-qubit-hole} and Fig.~\ref{fig:hole-encoding}.
In this scheme, a logical qubit is encoded using three holes (of different color).
For primal qubit, the $Z$ logical operator is the string encircling a hole and the logical $X$ operator is the string having tree like structure connecting the three holes.
For the dual qubit, the $X$ logical operator is the string encircling the hole and $Z$ logical operator is the string having tree like structure.
The idea of braiding holes around twists to implement CNOT gate was proposed in Ref.~\cite{Brown2017} in the context of surface codes.
In this paper, we explore the same idea in the context of color permuting twists.

\begin{figure*}
    \centering
    \begin{subfigure}{.45\textwidth}
        \centering
        \includegraphics[scale = .85]{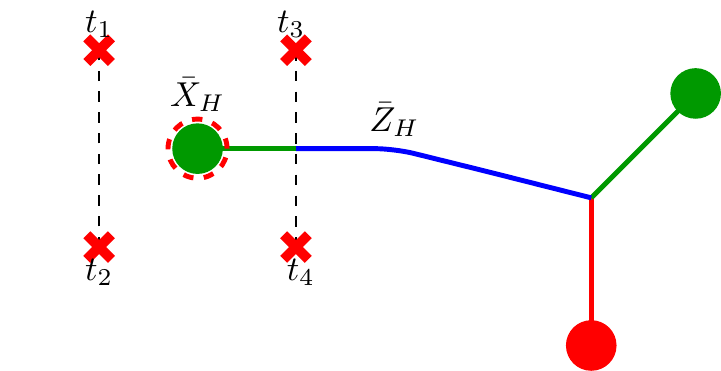}
        \subcaption{}
        \label{fig:twist-hole-braid-1}
    \end{subfigure}
    ~
    \begin{subfigure}{.45\textwidth}
        \centering
        \includegraphics[scale = .85]{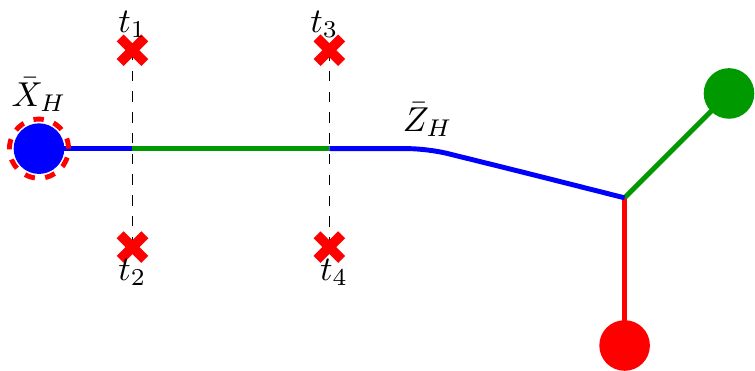}
        \subcaption{}
        \label{fig:twist-hole-braid-2}
    \end{subfigure}
    ~
    \begin{subfigure}{.45\textwidth}
        \centering
        \includegraphics[scale = .85]{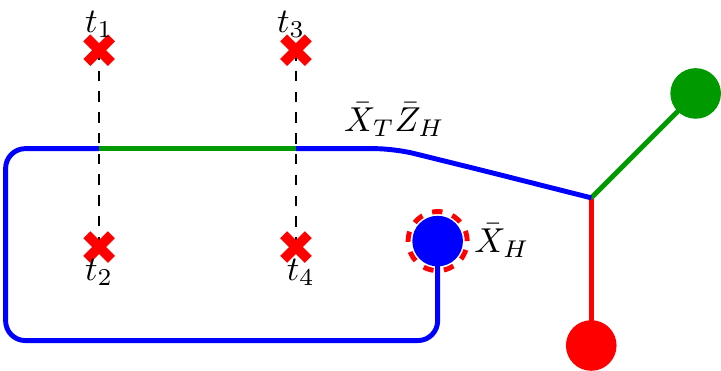}
        \subcaption{}
        \label{fig:twist-hole-braid-3}
    \end{subfigure}
    ~
    \begin{subfigure}{.45\textwidth}
        \centering
        \includegraphics[scale = .85]{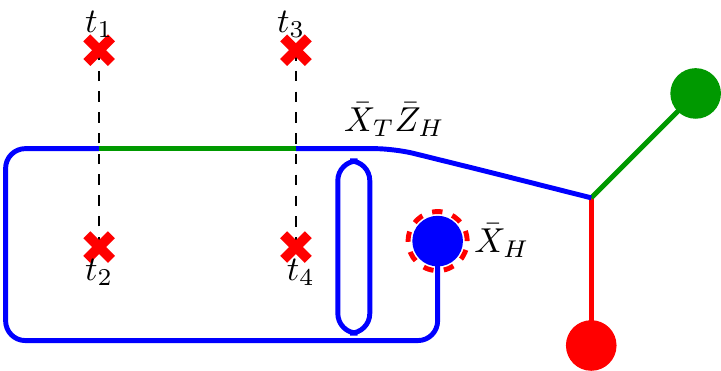}
        \subcaption{}
        \label{fig:twist-hole-braid-4}
    \end{subfigure}
    ~
    \begin{subfigure}{.45\textwidth}
        \centering
        \includegraphics[scale = .85]{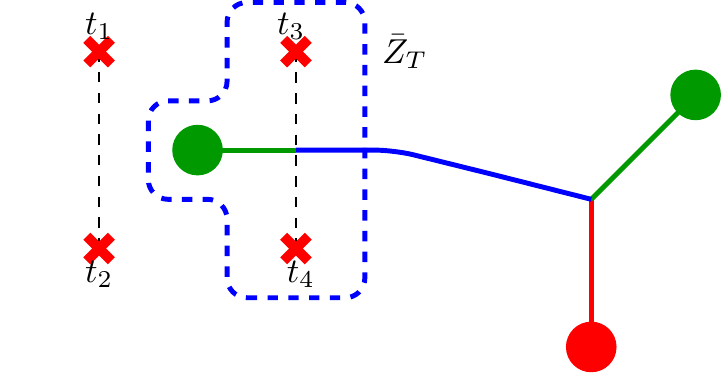}
        \subcaption{}
        \label{fig:twist-hole-braid-5}
    \end{subfigure}
     ~
    \begin{subfigure}{.45\textwidth}
        \centering
        \includegraphics[scale = .85]{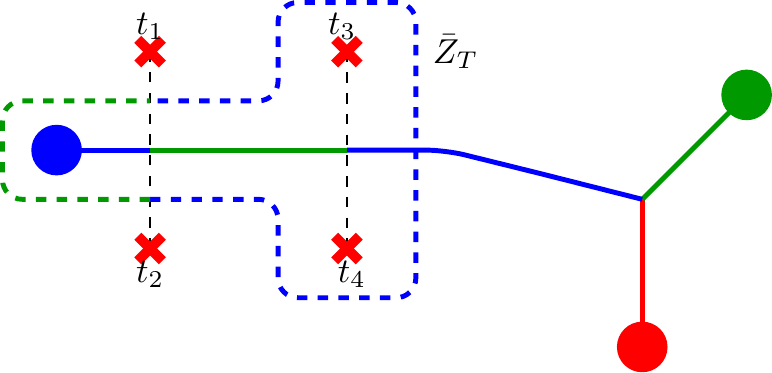}
        \subcaption{}
        \label{fig:twist-hole-braid-6}
    \end{subfigure}
     ~
    \begin{subfigure}{.45\textwidth}
        \centering
        \includegraphics[scale = .85]{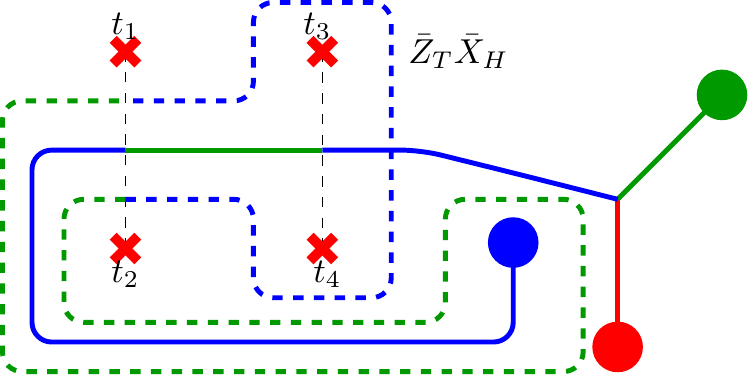}
        \subcaption{}
        \label{fig:twist-hole-braid-7}
    \end{subfigure}
    ~
    \begin{subfigure}{.45\textwidth}
        \centering
        \includegraphics[scale = .85]{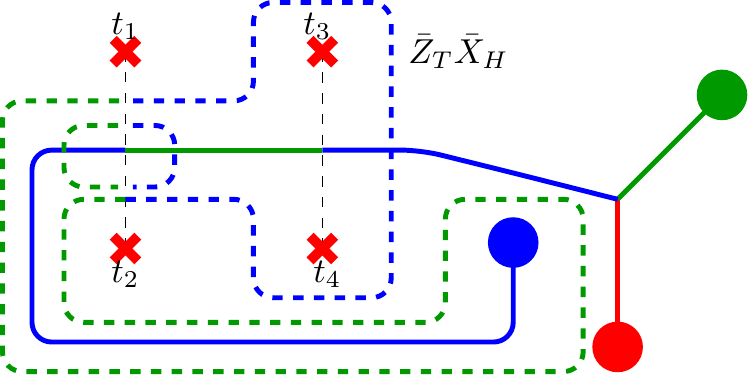}
        \subcaption{}
        \label{fig:twist-hole-braid-8}
    \end{subfigure}
    \caption{Braiding hole around twists $t_2$ and $t_4$ to implement the entangling gate given in Equation~\eqref{eqn:controlled-X}. (a) After crossing the $T$-line, the hole has changed its color from blue to green. Note that $\bar{X}$ does not change color because red twists do not permute the color of red strings. (b) Hole regains its color after crossing the second domain wall. Logical operators of the dual qubit encoded using holes are not transformed. (c) The hole is returned to its position after braiding. Note that the deformed logical $Z$ operator of dual qubit now encircles twists $t_2$ and $t_3$ and is equivalent to $\bar{Z}_H \bar{X}_T$. (d) The equivalence of the deformed logical $Z$ to $\bar{Z}_H \bar{X}_T$ can be seen by adding a stabilizer as shown. (e) Logical $Z$ of the qubit encoded using twists, $\Bar{Z}_T$, is deformed by hole movement as blue string cannot have terminal points on green hole. (f) $\Bar{Z}_T$ and hole change color as they crosses the domain wall of twists $t_1$ and $t_2$. Color of hole and $\Bar{Z}_T$ are different and hence $\Bar{Z}_T$ is further deformed. (g) The deformed logical $Z$ operator of twist qubit is equivalent to $\Bar{Z}_T \Bar{X}_H$. (h) By adding the stabilizer as shown, one can see the equivalence between deformed logical $Z$ operator to $\Bar{Z}_T \Bar{X}_H$.}
    \label{fig:twist-hole-braid}
\end{figure*}

\begin{figure*}
    \centering
    \begin{subfigure}{.45\textwidth}
        \centering
        \includegraphics[scale = .85]{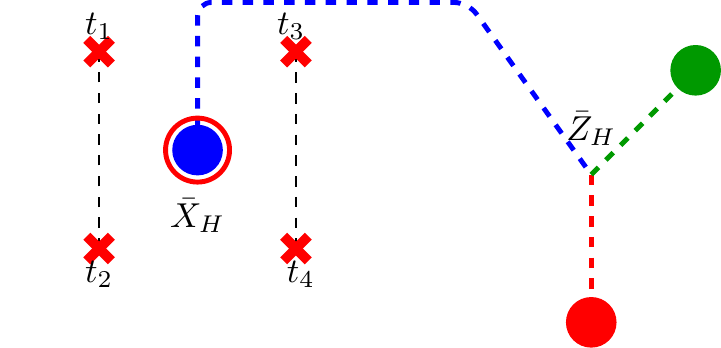}
        \subcaption{}
        \label{fig:hole-twist-braid-1}
    \end{subfigure}
    ~
    \begin{subfigure}{.45\textwidth}
        \centering
        \includegraphics[scale = .85]{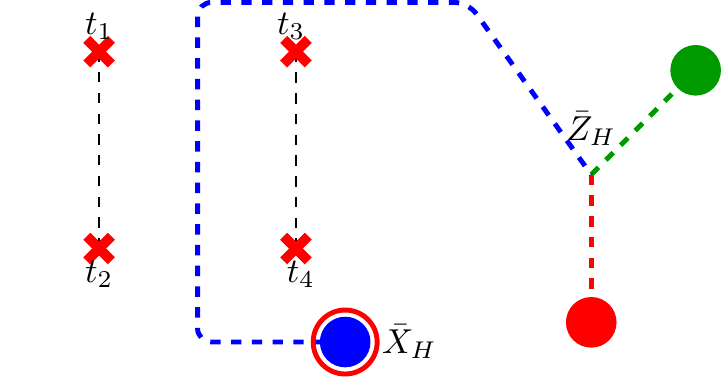}
        \subcaption{}
        \label{fig:hole-twist-braid-2}
    \end{subfigure}
    ~
    \begin{subfigure}{.45\textwidth}
        \centering
        \includegraphics[scale = .85]{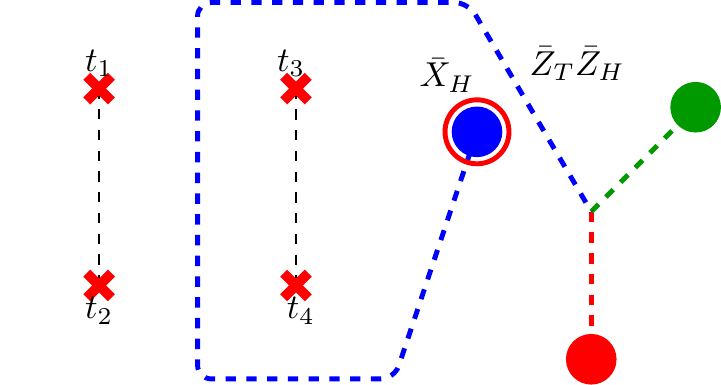}
        \subcaption{}
        \label{fig:hole-twist-braid-3}
    \end{subfigure}
    ~
    \begin{subfigure}{.45\textwidth}
        \centering
        \includegraphics[scale = .85]{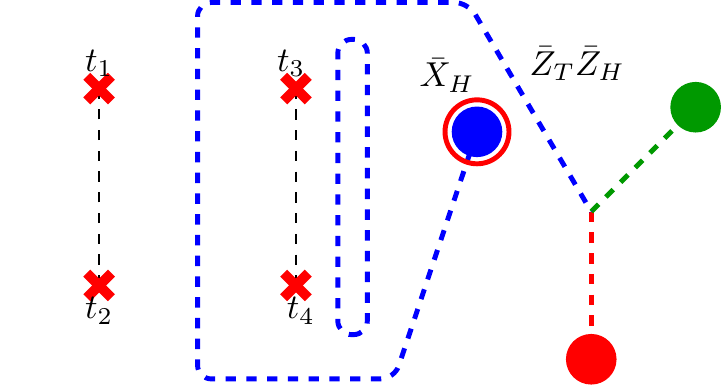}
        \subcaption{}
        \label{fig:hole-twist-braid-4}
    \end{subfigure}
     ~
    \begin{subfigure}{.45\textwidth}
        \centering
        \includegraphics[scale = .85]{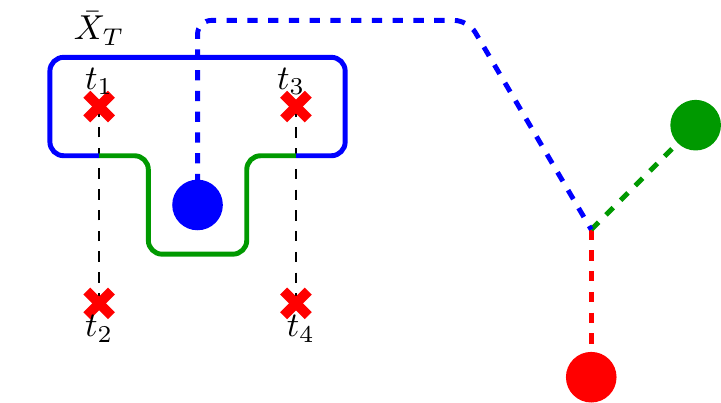}
        \subcaption{}
        \label{fig:hole-twist-braid-5}
    \end{subfigure}
     ~
    \begin{subfigure}{.45\textwidth}
        \centering
        \includegraphics[scale = .85]{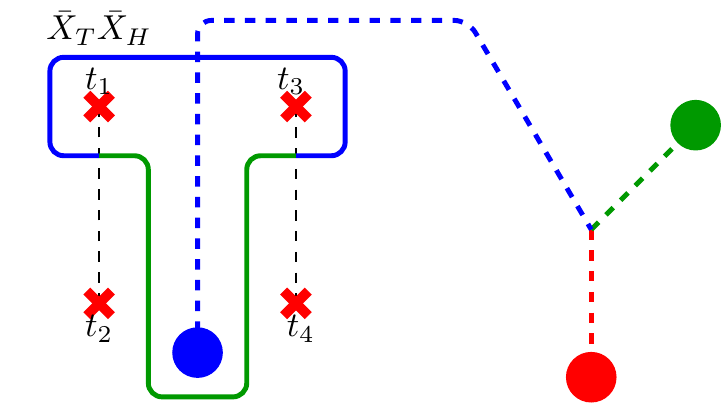}
        \subcaption{}
        \label{fig:hole-twist-braid-6}
    \end{subfigure}
     ~
    \begin{subfigure}{.45\textwidth}
        \centering
        \includegraphics[scale = .85]{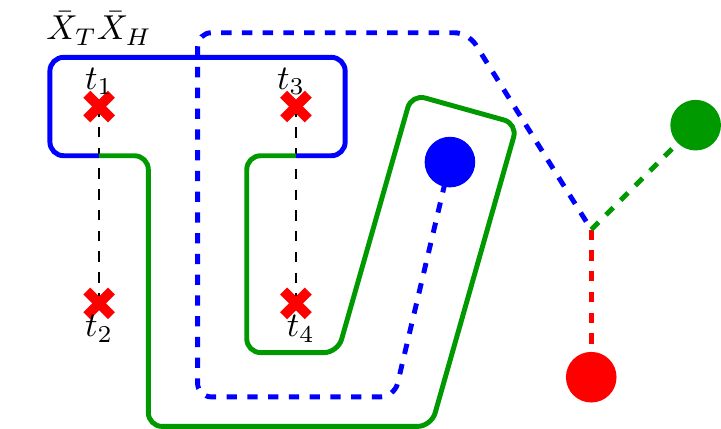}
        \subcaption{}
        \label{fig:hole-twist-braid-7}
    \end{subfigure}
    ~
    \begin{subfigure}{.45\textwidth}
        \centering
        \includegraphics[scale = .85]{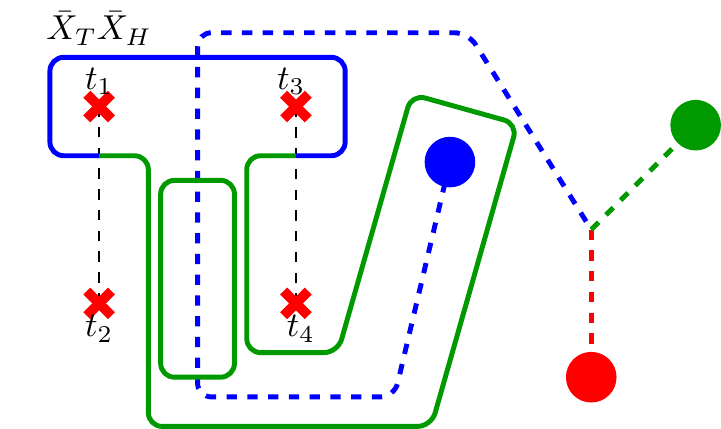}
        \subcaption{}
        \label{fig:hole-twist-braid-8}
    \end{subfigure}
\caption{Braiding hole around twists $t_3$ and $t_4$ to implement CNOT gate with qubit encoded using twists as control and qubit encoded using holes as target. (a) Logical $Z$ operator of the dual qubit is stretched as hole is braided around twists. (b) After braiding $\bar{Z}_H$ is deformed around twists $t_3$ and $t_4$ and encircles them. (c) The deformed logical $Z$ operator of the primal qubit is equivalent to $\bar{Z}_T \bar{Z}_H$. (d) The equivalence of the deformed $\bar{Z}_H$ to $\bar{Z}_T \bar{Z}_H$ can be seen by adding the stabilizer as shown. (e) The operator $\bar{X}_T$ deforms as hole is moved. This is because green string cannot have terminal points in a blue hole. (f) The operator $\bar{X}_T$ is dragged along with the hole and at the end of braiding, $\bar{X}_T$ encircles the hole. (g) The deformed $\bar{X}_T$ is equivalent to $\bar{X}_T \bar{X}_H$. (h) The equivalence stated in Fig.~\ref{fig:hole-twist-braid-7} can be seen by adding the stabilizer as shown.}
\label{fig:hole-twist-braid}
\end{figure*}

Consider the braiding as shown in Fig~\ref{fig:twist-hole-braid}.
Let $\Bar{Z}_{T} $ and $\Bar{X}_{T}$ be the $Z$ and $X$ logical operators of logical qubit encoded by twists and let $\Bar{Z}_{H} $ and $\Bar{X}_{H}$ be the $Z$ and $X$ logical operators of logical qubit encoded by holes.
The hole is braided around twists $t_2$ and $t_4$.
An interesting phenomenon occurs as the hole crosses the domain wall.
Suppose that the hole in question has blue color.
To move the hole across the domain wall, it has to pass through the same.
In doing so, we merge the hole with one of the blue modified faces.
Note that the outgoing edges from the modified faces are incident on green faces, see Fig.~\ref{fig:hex-lattice}.
When the hole is to be further moved from the domain wall, we now have to merge the modified face with one of the green faces.
When the hole is completely deformed away from the domain wall, its color will have permuted.
As the hole the crosses the first domain wall, the color changes from blue to green as shown in Fig.~\ref{fig:twist-hole-braid-1}.
Again when the hole crosses the domain wall second time, its color is changed back to blue as shown in Fig.~\ref{fig:twist-hole-braid-2}.
The hole is returned to its original position as shown in Fig.~\ref{fig:twist-hole-braid-3}.
The operator around twists $t_2$ and $t_4$ is $\bar{X}_T \bar{Z}_H$ which can be seen by adding a stabilizer as shown in Fig.~\ref{fig:twist-hole-braid-4} which decomposes the operator into $\bar{X}_T$ and $\bar{Z}_H$.
The transformation of the logical operators of the qubit encoded using twists is shown from Fig.~\ref{fig:twist-hole-braid-5} to Fig.~\ref{fig:twist-hole-braid-8}.
The logical operator transformation effected by this braiding is given in Table~\ref{tab:hole_twist_horizontal}.
From the transformations, it is clear that braiding hole around twists realizes an entangling gate. 

\begin{table}[htb]
    \centering
    \begin{tabular}{ccc}
    $\Bar{Z}_{T}$ & $\longrightarrow$ & $\Bar{Z}_{T}\Bar{X}_{H}$\\
    $\Bar{X}_{T}$ & $\longrightarrow$ & $\Bar{X}_{T}$ \\
    $\Bar{Z}_{H}$ & $\longrightarrow$ & $\Bar{X}_{T}\Bar{Z}_{H}$\\
    $\Bar{X}_{H}$ & $\longrightarrow$ & $\Bar{X}_{H}$
\end{tabular}
    \caption{Transformation of logical operators after braiding hole around twists as shown in Fig.~\ref{fig:twist-hole-braid}.}
    \label{tab:hole_twist_horizontal}
\end{table}
One can see that this transformation is equivalent to that obtained by applying Hadamard gate on the control qubit before and after CNOT gate.
This gate is used with the target qubit initialized in $|0\rangle$ to measure $X$ operator~\cite{Nielsen2010}.

\begin{lemma}[Twist-hole braiding-1]
Braiding hole encoding the dual qubit around twists $t_3$ and $t_4$ as shown in Fig.~\ref{fig:twist-hole-braid} realizes the gate 
\begin{equation}
    (H \otimes I) CNOT (H \otimes I) = |+\rangle \langle +| \otimes I + |-\rangle \langle -| \otimes X
    \label{eqn:controlled-X}
\end{equation}
with twist qubit as the control and dual qubit as the target.
\label{lm:twist-hole-braiding-1}
\end{lemma}

\begin{remark}
The string encircling twists $t_2$ and $t_4$  in Fig.~\ref{fig:Qb-cc-four-twist-encoding-color} is equivalent to logical $X$ operator.
\end{remark}

One can verify that by using primal qubit, i.e. the operator encircling the hole is $\Bar{Z}$ and the tree like operator is $\Bar{X}$, and braiding around $t_2$ and $t_4$ the following transformation is realized:
$\Bar{Z}_T \rightarrow \Bar{Z}_T \Bar{Z}_H $, $ \Bar{X}_T \rightarrow \Bar{X}_T$, $ \Bar{Z}_H \rightarrow \Bar{Z}_H$, $\Bar{X}_H \rightarrow \Bar{X}_T\Bar{X}_H$.
One can see that this is CNOT gate between primal qubit and qubit encoded with twists with the primal qubit as control.

To realize CNOT gate between qubits encoded using twists and holes, we use primal qubit.
We do Pauli frame update on the hole so that the logical $X$ and $Z$ are interchanged.
If the primal qubit encoded using holes is deformed around twists $t_3$ and $t_4$, then the transformation given in Table~\ref{tab:hole_twist_vertical} is realized.
The transformation of the logical operators of the primal qubit and the qubit encoded by twists is shown in Fig.~\ref{fig:hole-twist-braid}.
This is CNOT gate between qubit encoded by twist and primal hole with twist qubit as the control.
\begin{table}[htb]
    \centering
    \begin{tabular}{ccc}
    $\Bar{Z}_{T}$ & $\longrightarrow$ & $\Bar{Z}_{T}$\\
    $\Bar{X}_{T}$ & $\longrightarrow$ & $\Bar{X}_{T}\Bar{X}_H$ \\
    $\Bar{Z}_{H}$ & $\longrightarrow$ & $\Bar{Z}_T\Bar{Z}_{H}$\\
    $\Bar{X}_{H}$ & $\longrightarrow$ & $\Bar{X}_{H}$
\end{tabular}
\caption{Transformation of logical operators after braiding hole around twists as shown in Fig.~\ref{fig:hole-twist-braid}.}
\label{tab:hole_twist_vertical}
\end{table}

\begin{lemma}[Twist-hole braiding-2]
Performing Pauli update $\bar{Z} \leftrightarrow \bar{X}$ on the primal qubit and braiding hole encoding primal qubit around twists $t_3$ and $t_4$ as shown in Fig.~\ref{fig:hole-twist-braid} realizes CNOT gate with qubit encoded using twists as control and qubit encoded using holes as target.
\label{lm:twist-hole-braiding-2}
\end{lemma}

\begin{remark}
The string encircling twists $t_3$ and $t_4$ in Fig.~\ref{fig:Qb-cc-four-twist-encoding-color} is equivalent to logical $Z$ operator.
\end{remark}

Similarly, by using primal qubit encoded using holes and braiding around $t_3$ and $t_4$, one can realize controlled phase gate between twist and hole with twist as the control.

\begin{figure*}
\centering
    \begin{subfigure}{.45\textwidth}
          \centering
          \includegraphics[scale = .85]{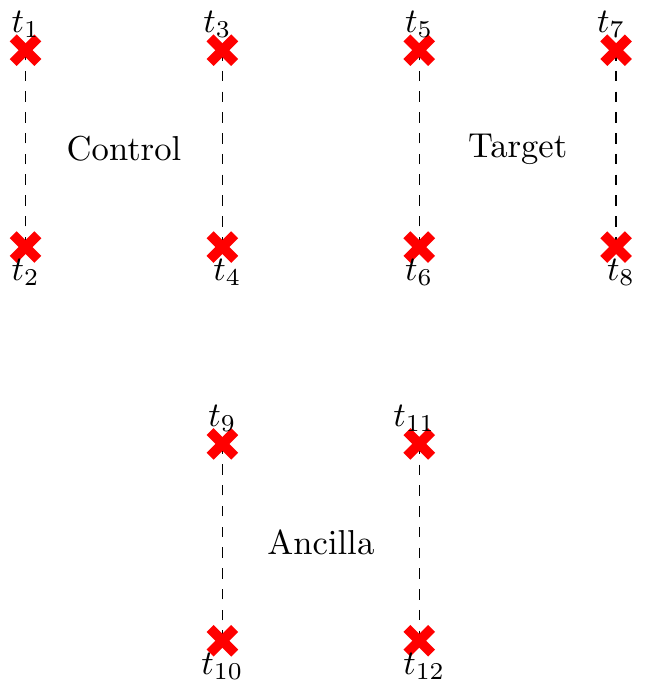}
          \subcaption{}
          \label{fig:color-perm-CNOT-1}
    \end{subfigure}
    ~
    \begin{subfigure}{.45\textwidth}
          \centering
          \includegraphics[scale = .85]{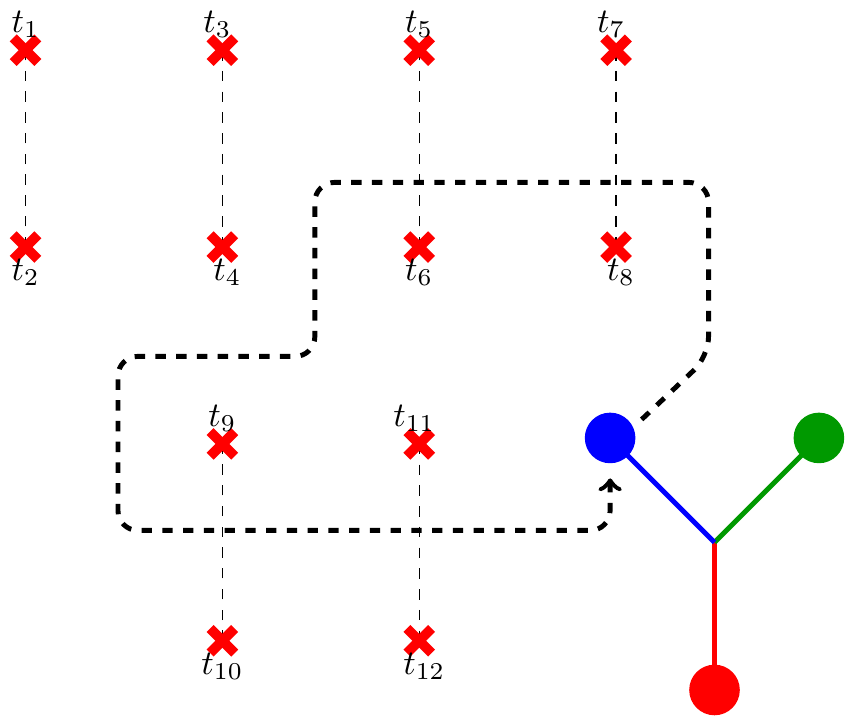}
          \subcaption{}
          \label{fig:color-perm-CNOT-2}
    \end{subfigure}
    ~
    \begin{subfigure}{.45\textwidth}
          \centering
          \includegraphics[scale = .8]{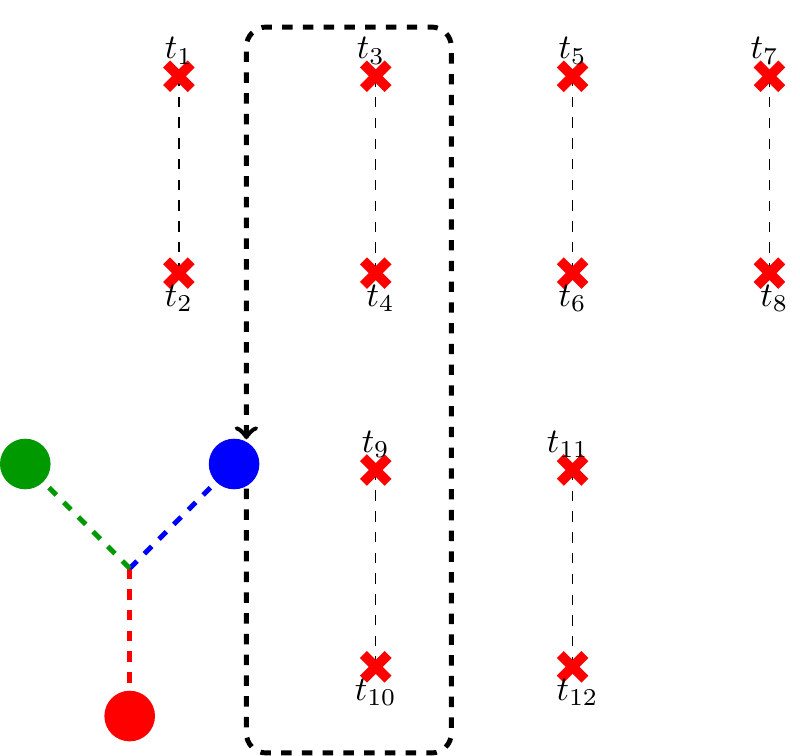}
          \subcaption{}
          \label{fig:color-perm-CNOT-3}
    \end{subfigure}
    ~
    \begin{subfigure}{.45\textwidth}
          \centering
          \includegraphics[scale = .8]{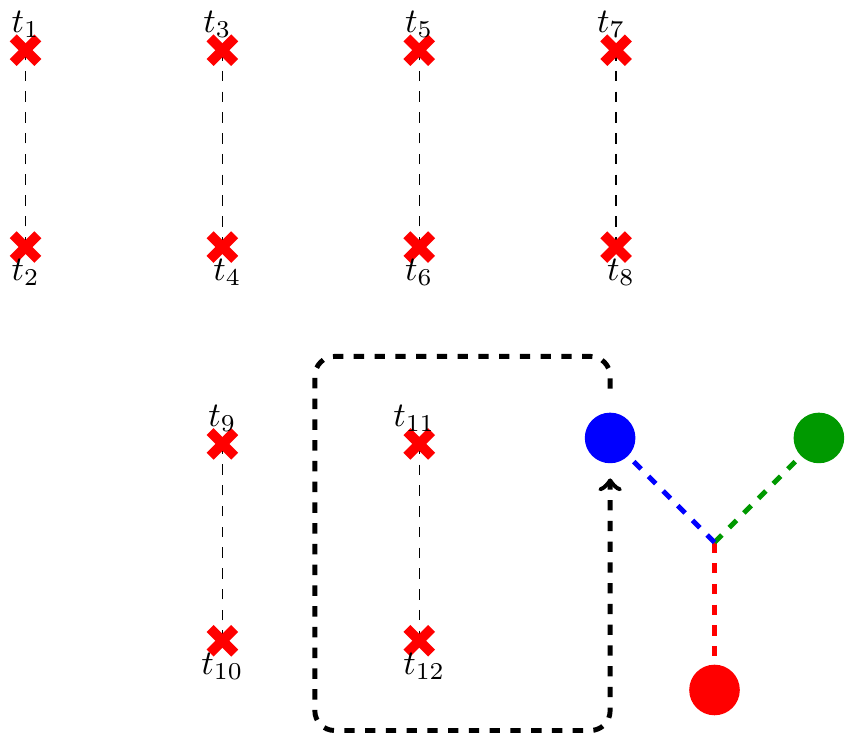}
          \subcaption{}
          \label{fig:color-perm-CNOT-4}
    \end{subfigure}
    \caption{Protocol for implementing CNOT gate using twists and holes. (a) This protocol requires an ancilla to implement CNOT gate. Twists used for encoding control (C), target (T) and ancilla (A) qubits are shown. Logical $Z$ and $X$ operators are in the canonical form. (b) Braiding hole along the operator $\Bar{X}_T \Bar{X}_A$ and measuring it in the end to implement joint $X$ measurement on target and ancilla qubits. (c) Joint $Z$ measurement on control and ancilla qubits is performed by braiding hole along the string $\Bar{Z}_C \Bar{Z}_A$ and measuring it. (d) Ancilla qubit is measured by braiding hole along $\Bar{X}_A$ which is the string encircling twists $t_{11}$ and $t_{12}$  after Pauli frame update measuring it.}
    \label{fig:color-perm-CNOT}
\end{figure*}

\subsection{CNOT Gate} 
We now present the encoded version of the CNOT gate protocol given in Table~\ref{tab:cc_cnot}. 
This protocol makes use of the encoded gates realized using twists and holes as discussed in the previous subsection.

\begin{table}[htb]
    \centering
    \begin{tabular}{ll}
    \hline
    & Protocol for implementing Encoded CNOT gate \\
    \hline 
    (1) & Prepare ancilla encoded using twists in the state $|\Bar{0}\rangle$. \\
    (2) & Perform $\Bar{X}_T \Bar{X}_A$ parity measurement on ancilla and target qubit \\
    & by braiding hole around twists encircled by the operator $\Bar{X}_T \Bar{X}_A$.\\
    (3) & Perform  $\Bar{Z}_C \Bar{Z}_A$ parity measurement on ancilla and control qubit\\
        & by braiding hole around twists encircled by  the operator $\Bar{Z}_C \Bar{Z}_A$.\\
    (4) & Perform Hadamard on ancilla qubit and measure it in $X$ basis.\\
    (5) & Measure the ancillas encoded using holes and apply\\
     & correction as given in Equation~\ref{eqn:CNOT}.\\
    \hline
    \end{tabular}
    \caption{Protocol for implementing encoded CNOT gate by joint measurements with an ancilla.}
    \label{tab:cc_cnot_encoded}
\end{table}

\begin{remark}
Fresh ancillas are prepared for steps $2$ to $4$ in the protocol given in Table~\ref{tab:cc_cnot_encoded} for encoded CNOT gate.
\end{remark}

The procedure for implementing the encoded CNOT gate using color permuting twists is shown in Fig.~\ref{fig:color-perm-CNOT}.
In Fig.~\ref{fig:color-perm-CNOT-1}, control, target and ancilla qubits encoded using a quadruple of twists is shown.
We abbreviate control, target and ancilla qubits as C, T and A respectively.
The ancilla is initialized in the $|\Bar{0}\rangle$ state.
Joint $\Bar{X}_T \Bar{X}_A$ parity measurement by using dual qubit encoded using holes is shown in Fig.~\ref{fig:color-perm-CNOT-2}.
This parity measurement is equivalent to doing the gate indicated in Equation~\eqref{eqn:controlled-X} with target and ancilla $A$ qubits as control and the hole as target.
Recall that braiding a dual qubit around twists encircled by logical $X$ operator realizes the gate in Equation~\eqref{eqn:controlled-X} with twist qubit being control and dual qubit being target.
Braiding hole as shown in Fig.~\ref{fig:color-perm-CNOT-2} realizes the gate in Equation~\eqref{eqn:controlled-X} between target qubit and ancilla qubit encoded using twists as control and qubit encoded using holes as target.
The dual qubit is initialized in the  $|\bar{0}\rangle$ state, braided around twists as shown and measured in the end.
The measurement outcome reveals the parity of the operator $\Bar{X}_T \Bar{X}_A$.
Similarly, to carry out $\Bar{Z}_C \Bar{Z}_A$ parity measurement, a primal qubit encoded by holes is created and its logical $X$ and $Z$ operators are interchanged by Pauli frame update. 
The primal qubit is initialized in the $+1$ eigenstate of $\Bar{X}_H$ so that after Pauli frame update, the state of the primal qubit is $|\bar{0}\rangle$, deformed around twists and measured in the end, see Fig.~\ref{fig:color-perm-CNOT-3}.
This braiding realizes CNOT gate between control and ancilla qubits encoded using twists as control and the dual qubit encoded with holes as target.
Finally, Hadamard on ancilla is performed by Pauli frame update and measurement in the $X$ basis is done as shown in Fig.~\ref{fig:color-perm-CNOT-4}.

We summarize the results of this section in the theorem below.
\begin{theorem}
Using color permuting twists, single qubit Clifford gates are realized by Pauli frame update and CNOT gate by joint parity measurements with an ancilla.
\label{thm:color-perm-gates}
\end{theorem}

So far we have discussed implementation of Clifford gates.
To achieve universality, we need a non-Clifford gate.
Non-Clifford gate is implemented using the technique of magic state distillation.
A discussion on the implementation of non-Clifford gate is presented in Appendix~\ref{sec:non-clifford}.
The proposed protocol  is a minor variation of the one given in Ref.~\cite{Bravyi2006}.

\section{Conclusion}
\label{sec:conclusion}
In this paper, we have presented a systematic way to introduce charge permuting and color permuting twists in $2$-colexes.
We have also discussed the coding theoretic aspects of the same.
We have shown that Clifford gates can be implemented using charge permuting twists by braiding alone.
Encoded gates in the case of color permuting twists are implemented by a combination of Pauli frame update for single qubit gates and joint parity measurements with an ancilla for CNOT gate.
Joint measurements are carried out by braiding holes around twists.
To implement non-Clifford gate, we use magic state injection.
A direction for future research could be to use techniques other than magic state distillation such as the one presented in Ref.~\cite{Brown2020} which could be used in the case of color codes with twists.
Decoding color codes with twists is another fruitful area for further exploration.

\section*{Acknowledgement}
We thank the reviewers for their comments which helped in improving the presentation of the paper.
This research was supported by the Science and Engineering Research Board, Department of Science and Technology under Grant No. EMR/2017/005454.

\section*{Appendix}
\appendix
\section{Stabilizer commutation for charge permuting twists}
\label{sec:charge-stabilizer}
As mentioned before, we have three types of faces in the lattice: twists ($\tau$), faces through which domain wall passes through ($f_D$) and faces that are neither twists nor through which domain wall passes through ($f_{\Bar{D}}$).
When any two faces are not adjacent, then they do not share common vertices and hence the stabilizers defined on them commute.
The stabilizer commutation in the case of adjacent faces needs to be checked.
We have to check commutation between the six cases as listed below.

\begin{compactenum}[a)]
    \item \emph{$\tau$ and $\tau$}: 
    When the two twists are adjacent, then they share exactly an edge.
    The Pauli operators corresponding to the two stabilizers are same on each of the common vertices.
    Hence, commutation follows.
    \item \emph{$\tau$ and $f_D$}: These faces share exactly an edge.
    The Pauli operators corresponding to the two stabilizers differ at each common vertex ($X$ on twist and $Y$ and $Z$ on $f_D$ in the case of $X$ twists).
    Hence, commutation follows.
    \item \emph{$\tau$ and $f_{\Bar{D}}$}: The Pauli operators for the twist stabilizer and the normal face stabilizers differ on the common vertices shared between twist and normal faces, . Hence, they commute.
    \item \emph{$f_D$ and $f_D$}: The Pauli operators corresponding to the two stabilizers, $B_{f,1}$ and $B_{f,2}$ are different at each vertex. 
    Note that all the faces have an even number of edges and hence these stabilizer generators anti-commute an even number of times. 
    Therefore, these stabilizers commute.
    \item \emph{$f_D$ and $f_{\Bar{D}}$}: Along the common edge to these faces, the Pauli operators are either same or different. Hence, commutation follows.
    \item \emph{$f_{\Bar{D}}$ and $f_{\Bar{D}}$}: Stabilizers commute as the Pauli operators along the common edge are either same or different.
\end{compactenum}

\onecolumngrid
\section{Braiding of charge permuting twists}
\label{sec:charge-braiding}
In this section, we show that the product of  $\bar{X}_1$ and $\bar{Z}_1$ as shown in Fig.~\ref{fig:equivalent_LO} is equivalent to logical $\bar{Y}_1$ operator up to gauge, see Fig.~\ref{fig:Qb-cc-logY}.
For convenience we drop the subscripts.

\begin{figure}[htb]
    \centering
    \begin{subfigure}{.45\textwidth}
        \centering
        \includegraphics[scale = .75]{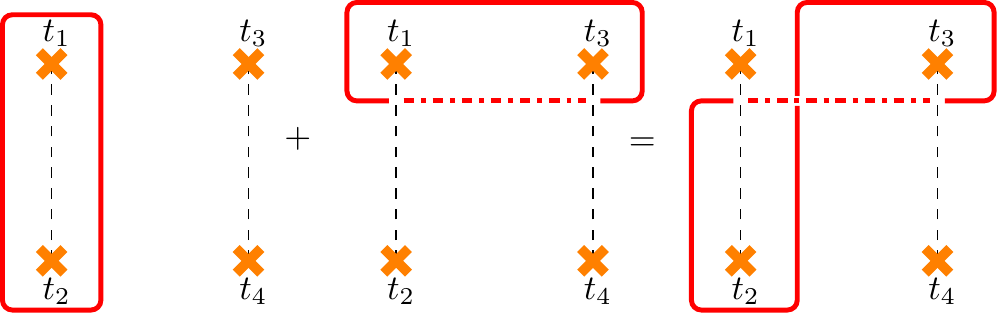}
        \subcaption{}
        \label{fig:Qb-cc-logY-3}
    \end{subfigure}
    ~
    \begin{subfigure}{.45\textwidth}
        \centering
        \includegraphics[scale = .75]{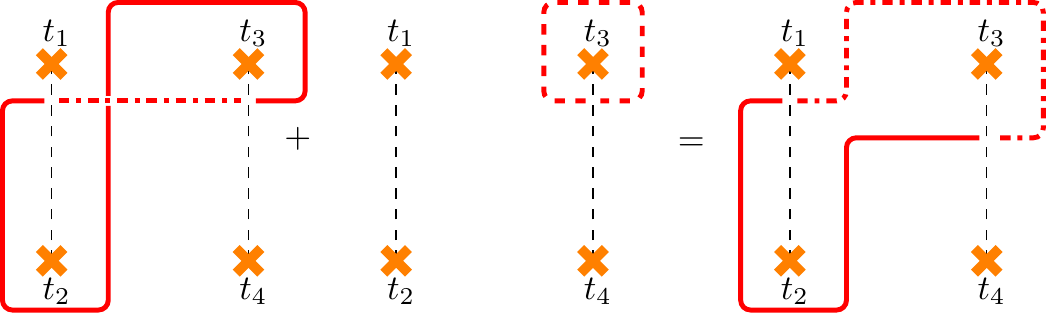}
        \subcaption{}
        \label{fig:Qb-cc-logY-4}
    \end{subfigure}
    \caption{The canonical $\Bar{Y}$ which is the product of $\Bar{Z}$ and $\Bar{X}$ is equivalent to the string operator shown in Fig.~\ref{fig:Qb-cc-logY-4} up to a gauge operator. Note that the logical $X$ and $Z$ in Fig.~\ref{fig:Qb-cc-logY-3} are not canonical. (a) The logical $Z$ are $X$ operators obtained by adding gauge operators are combined. The resulting string encircles twists $t_2$ and $t_3$ and crosses itself once. (b) The self crossing string is brought to the canonical form by adding the stabilizer defined on twist face $t_3$.}
    \label{fig:Qb-cc-logY}
\end{figure}

The deformation of logical $X$ operator to logical $Y$ operator after braiding twists $t_1$ and $t_2$ is shown in Fig.~\ref{fig:cp}.
The string in Fig.~\ref{fig:Qb-cc-charge-Log-X-deformation-9} is logical $Y$ operator up to a gauge operator. 

\begin{figure}[htb]
    \centering
    \begin{subfigure}{.225\textwidth}
        \centering
        \includegraphics[scale = .75]{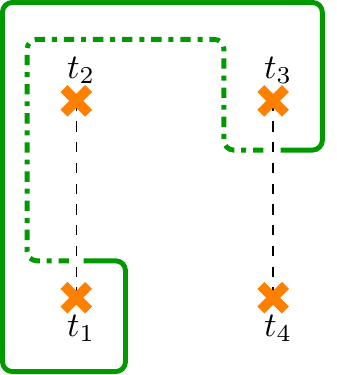}
        \subcaption{}
        \label{fig:Qb-cc-charge-Log-X-deformation-2}
    \end{subfigure}
    ~
    \begin{subfigure}{.225\textwidth}
        \centering
        \includegraphics[scale = .75]{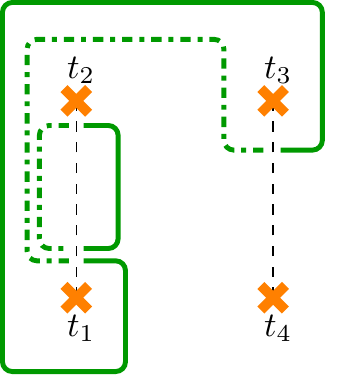}
        \subcaption{}
        \label{fig:Qb-cc-charge-Log-X-deformation-3}
    \end{subfigure}
    ~
    \begin{subfigure}{.225\textwidth}
        \centering
        \includegraphics[scale = .75]{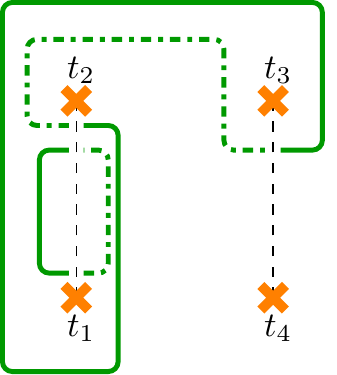}
        \subcaption{}
        \label{fig:Qb-cc-charge-Log-X-deformation-4}
    \end{subfigure}
    ~
    \begin{subfigure}{.225\textwidth}
        \centering
        \includegraphics[scale = .75]{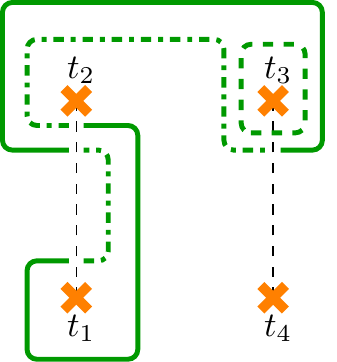}
        \subcaption{}
        \label{fig:Qb-cc-charge-Log-X-deformation-5}
    \end{subfigure}
    ~
    \begin{subfigure}{.225\textwidth}
        \centering
        \includegraphics[scale = .75]{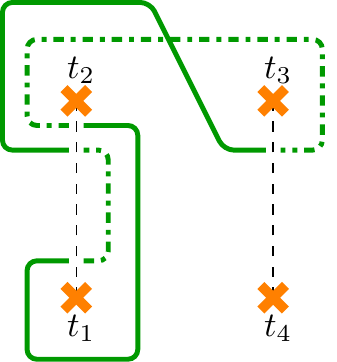}
        \subcaption{}
        \label{fig:Qb-cc-charge-Log-X-deformation-6}
    \end{subfigure}
    ~
    \begin{subfigure}{.225\textwidth}
        \centering
        \includegraphics[scale = .75]{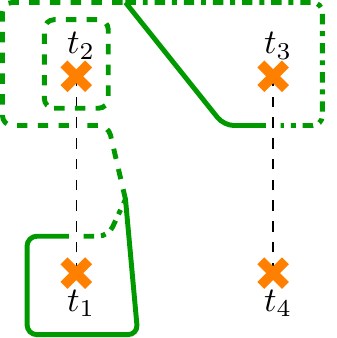}
        \subcaption{}
        \label{fig:Qb-cc-charge-Log-X-deformation-7}
    \end{subfigure}
     ~
    \begin{subfigure}{.225\textwidth}
        \centering
        \includegraphics[scale = .75]{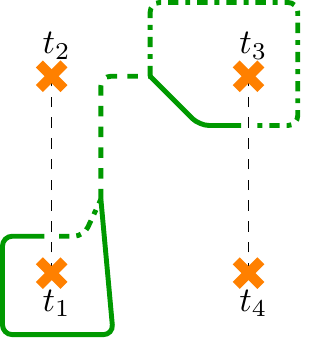}
        \subcaption{}
        \label{fig:Qb-cc-charge-Log-X-deformation-8}
    \end{subfigure}
         ~
    \begin{subfigure}{.225\textwidth}
        \centering
        \includegraphics[scale = .75]{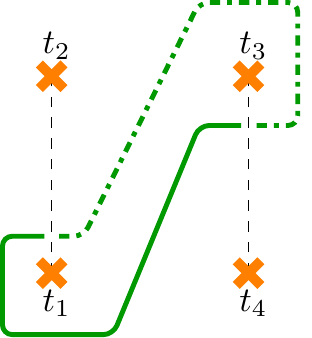}
        \subcaption{}
        \label{fig:Qb-cc-charge-Log-X-deformation-9}
    \end{subfigure}
\caption{Deforming logical $X$ operator after braiding twists $t_1$ and $t_2$ to logical $Y$ operator. The logical $X$ operator after braiding is deformed as shown by adding stabilizers. (a) On braiding twists $t_1$ and $t_2$, the logical $X$ operator is deformed as shown. (b) We add the stabilizer of the faces between domain wall to deform the string shown in Fig.~\ref{fig:Qb-cc-charge-Log-X-deformation-2}. (c) We now add the other stabilizer of faces between domain wall to the string obtained from Fig~\ref{fig:Qb-cc-charge-Log-X-deformation-3}. This stabilizer has Pauli $Z$ and Pauli $Y$ on the left and right of the domain wall respectively. (d) Upon adding the stabilizers of faces between domain wall, a portion of the string is deformed to the right of domain wall. We now add the stabilizer of twist $t_3$. (e) Adding twist stabilizer to the string obtained from Fig.~\ref{fig:Qb-cc-charge-Log-X-deformation-4}, we obtain the string as shown. The $Z$ and $Y$ portion of the string around $t_3$ is interchanged and the string crosses itself. (f) The $Z$ and $Y$ portions of the string are combined to form $X$ string. To this, we add the stabilizer of twist $t_2$ braiding the string completely to the right side of the domain wall. (g) The deformed string as a result of adding the stabilizer of twist $t_2$. Note that the $X$ string is now on the right side of the domain wall. (h) The $X$ string is decomposed into $Z$ and $Y$ strings.}
\label{fig:cp}
\end{figure}
\medskip
\twocolumngrid

\section{Existence of T-line}
\label{sec:t-line}
\begin{lemma}[Existence of $T$-lines]
Every pair of $c$-color permuting twists created together are connected by a $T$-line.
\label{lm:t-lines}
\end{lemma}

\begin{proof}
Without loss of generality, assume that the twist faces have red color i.e. the twists permute the color of blue and green anyons.
Consider a twist face in the pair, call it $\tau_1$.
Since twists are faces with an odd number of edges, it is not possible to color faces around them consistently with two colors.
Therefore, there exists a pair of faces $f_{1,1}$ and $f_{1,2}$ which are adjacent and share the same color.
Also, from the construction, one of these faces, say, $f_{1,2}$ is a modified face.
The common edge to faces $f_{1,1}$ and $f_{1,2}$ is red and is incident on a normal red face, call it $g_1$.
Note that $ f_{1,1}, f_{1,2} \in A_2(g_1)$.
As a consequence, faces around $g_1$ also cannot be colored consistently with two colors.
Therefore, there exists a face $f_{2,1} \in A_2(g_1)$ such that $c(f_{2,1}) = c(f_{1,1})$.
The face $f_{2,1}$ can be adjacent to either $f_{1,1}$ or $f_{1,2}$.
Assuming that it is adjacent to $f_{1,2}$, we now have another pair of faces of the same color adjacent to each other.
Continuing this argument by taking $c$-colored faces connected by a $c$-colored edge, we establish the existence of $T$-lines in the presence of color permuting twists.
\end{proof}

\section{Non-Clifford gate}
\label{sec:non-clifford}
In this section we discuss the implementation of a non-Clifford gate.
For both charge permuting and color permuting twists we use magic state distillation, see
Table~\ref{tab:non-clifford}, to implement the non-Clifford gate.
We have dropped the joint $Z$ measurement in the protocol given in Ref.~\cite{Bravyi2006}. 
Similarities and differences with previous protocols is given in Table~\ref{tab:comp_protocol}.

\begin{table}[htb]
    \begin{tabular}{ll}
    \hline
    & Protocol for $\Lambda\left( e^{i \theta} \right)$ gate\\
    \hline
        (0) & Prepare the magic state $|\phi\rangle = 2^{-1/2}(|0\rangle + e^{i \theta}|1\rangle)$. \\
        
        (1) & Perform CNOT gate between data qubit and magic state\\
            & with the former as control and the latter as target.\\
        
        (2) & Measure the magic state qubit. \\
        (3) & If $\theta = \pi / 4$, apply phase gate when measurement outcome is $1$.\\
        \hline
\end{tabular}
\caption{Protocol for implementing non-Clifford gate.}
\label{tab:non-clifford}
\end{table}

\begin{table}[htb]
    \centering
    \begin{tabular}{p{3cm}|p{1.5cm}|p{1.5cm}|p{1.5cm}}
    \hline\hline
    Step     & Ref.~\cite{Bravyi2005} & Ref.~\cite{Bravyi2006}& Proposed protocol\\
    \hline
    Joint $Z$ parity measurement &  present & present & absent\\
    \hline
    CNOT gate &  present & present & present\\
    \hline
    Ancilla measurement &  absent & present & present\\
    \hline
    \end{tabular}
    \caption{Comparison of proposed protocol with some prior work.}
    \label{tab:comp_protocol}
\end{table}

Let $|\psi \rangle = a|0\rangle + b|1\rangle$ be the state of the data qubit on which the non-Clifford gate has to be performed. 
The state of ancillary qubit prepared in magic state is $|\phi \rangle = 2^{-1/2}(|0\rangle + e^{i \theta}|1\rangle)$.
After implementing CNOT gate, the joint state $\ket{\psi}\ket{\phi}$ is transformed to $|\beta \rangle $ where 
	\begin{eqnarray*}
	 |\beta \rangle &=& \frac{1}{\sqrt{2}}\left(a|00\rangle + a e^{i \theta}|01\rangle + b|11\rangle + b e^{i \theta}|10\rangle \right) \\
	                &=& \frac{1}{\sqrt{2}}\left(\left( a|0\rangle + be^{i \theta}|1\rangle \right)  |0\rangle + e^{i \theta}\left( a|0\rangle + be^{-i \theta}|1\rangle \right) |1\rangle \right) \\
	                &=& \frac{1}{\sqrt{2}}\left(\Lambda\left( e^{i \theta} \right)|\psi \rangle |0\rangle + e^{i \theta} \Lambda\left( e^{-i \theta} \right)|\psi \rangle  |1\rangle \right).
	\end{eqnarray*}
Upon measuring ancillary qubit, either the desired gate or its conjugate is implemented on the data qubit with equal probability. 
For $\theta = \pi / 4$, if the measurement outcome is $1$, then the desired transformation on data qubit can be obtained by performing phase gate.

Magic state distillation protocols can be found in Refs.~\cite{Bravyi2012, Jones2013, Haah2017, Chamberland2020}.
Twists in color codes can be initialized in magic state by state injection~\cite{Hastings2015}.
The color code lattice with the encoded magic state is merged with the main lattice by lattice surgery~\cite{Horsman2012}.
Then, the non-Clifford gate is implemented by performing the encoded versions of operations described in the protocol.
After measuring the encoded ancilla qubit i.e. qubit initialized in the magic state, the sublattice containing the encoded ancilla qubit can be detached.

\end{document}